\documentclass[11pt]{article} 

\usepackage{amsmath,amssymb,amsthm,amsfonts,latexsym,bbm,xspace,graphicx,float,mathtools}
\usepackage[linktocpage=true,colorlinks,citecolor=blue,linkcolor=red]{hyperref}
\usepackage{color}
\usepackage{algorithm}
\usepackage{algorithmic}

\newcommand{\correct}{\textsc{Correct} }
\newcommand{\inc}{\textsc{Incorrect} }

\usepackage{geometry} 
\usepackage{graphicx} 

\usepackage{array} 
\usepackage{paralist} 
\usepackage{verbatim} 
\usepackage{subfig} 

\newtheorem{theorem}{Theorem}

\newtheorem{corollary}{Corollary}
\newtheorem{lemma}{Lemma}

\newtheorem{proposition}{Proposition}
\newtheorem{fact}{Fact}
\newtheorem{definition}{Definition}

\newtheorem*{conjecture*}{Conjecture}

\begin{document}

\title{Asynchronous Majority Dynamics in Preferential Attachment Trees}
\author{Maryam Bahrani\footnote{Princeton University, mbahrani@alumni.princeton.edu.} 
\and
Nicole Immorlica\footnote{Microsoft Research, nicimm@microsoft.com} 
\and
Divyarthi Mohan\footnote{Princeton University, dm23@cs.princeton.edu} 
\and
S. Matthew Weinberg\footnote{Princeton University, smweinberg@princeton.edu. Supported by NSF CCF-1717899.}
}
\date{}

\newcommand{\E}{\mathbb{E}}
\maketitle
\begin{abstract}
We study information aggregation in networks where agents make binary decisions (labeled incorrect or correct). Agents initially form independent private beliefs about the better decision, which is correct with probability $1/2+\delta$. The dynamics we consider are asynchronous (each round, a single agent updates their announced decision) and non-Bayesian (agents simply copy the majority announcements among their neighbors, tie-breaking in favor of their private signal). 

Our main result proves that when the network is a tree formed according to the preferential attachment model~\cite{BarabasiA99}, with high probability, the process stabilizes in a correct majority within $O(n\log n/\log\log n)$ rounds. We extend our results to other tree structures, including balanced $M$-ary trees for any $M$.


\end{abstract}
\addtocounter{page}{-1}
\thispagestyle{empty}
\newpage


\section{Introduction}

Individuals form opinions about the world both through private investigation and through discussion with one-another.  A citizen, trying to decide which candidate's economic policies will lead to more jobs, might form an initial belief based on her own employment history. However, her stated opinion might be swayed by the opinions of her friends.  The dynamics of this process, together with the social network structure of the individuals, can result in a variety of societal outcomes.  Even if individuals are well-informed, i.e., are more likely to have correct than incorrect initial beliefs, certain dynamics and/or network structures can cause large portions of the population to form mistaken opinions. 

A substantial body of work exists modeling these dynamics mathematically, which we overview in Section~\ref{sec:related}. This paper focuses on the model of {\em asynchronous majority dynamics}.  Initially, individuals have private beliefs over a binary state of the world, but no publicly stated opinion.  Initial beliefs are independent: \correct with probability $1/2+\delta$, and \inc with probability $1/2-\delta$. In each time step, a random individual is selected to announce a public opinion. Each time an individual announces a public opinion, they simply copy the majority of their neighbors' announcements, tie-breaking in favor of their private belief. This is clearly naive: a true Bayesian would reason about the redundancy of information among the opinions of her friends, for example.  Majority (or other non-Bayesian) dynamics are generally considered a more faithful model of agents with bounded rationality (e.g. voters), whereas Bayesian dynamics are generally considered a more faithful model of fully rational actors (e.g. financial traders). We consider asynchronous announcements\footnote{Unlike synchronous models where all agents announce simultaneously.} which are a more faithful model of human decisions (e.g. citizens deciding which candidate is better).

It's initially tempting to conjecture that these dynamics in a connected network should result in a \correct consensus; after all, the majority is initially \correct (with high probability) by assumption. Nonetheless, it's well-understood that individuals can fail miserably to learn. Suppose for instance that the individuals form a complete graph. Then in asynchronous majority dynamics, whichever individual is selected to announce first will have their opinion copied by the entire network. As this opinion is \inc with constant probability, there's a good chance that the entire network makes the wrong decision (this is known as an \emph{information cascade}, and is not unique to asynchonous majority dynamics~\cite{Banerjee92,BikchandaniHW92}). So the overarching goal in these works is to understand in \emph{which} graphs the dynamics stabilize in correctness with high probability.

For most previously studied dynamics (discussed in Section~\ref{sec:related}), ``correctness'' means a \correct consensus. This is because the models terminate in a consensus with probability $1$, and the only question is whether this consensus is correct or not. With majority dynamics, it is certainly possible that the process stabilizes without a consensus. To see this, suppose individuals form a line graph.  In this case, two adjacent individuals with the same initial belief are likely to form a ``road block'' (if both announce before their other neighbors), sticking to their initial beliefs throughout the process. In this case, with high probability a constant fraction of individuals terminate with a \correct opinion, but also a constant fraction terminate with an \inc opinion. As consensus is no longer guaranteed, we're instead interested in understanding network structures for which the dynamics converge, with high probability, to a \emph{majority} of nodes having the \correct opinion (i.e., if a majority vote were to be taken, would it be correct w.h.p.?). 

Prior work shows that, to reach a \correct consensus, it's sufficient for the social network to be sparse (every individual has only a constant number of neighbors) and expansive (every group of individuals have many friends outside the group)~\cite{FeldmanILW14}, and the tools developed indeed make strong use of both assumptions.  Many networks of interest, however, like the hierarchy of employees in a corporation, are neither sparse nor expansive. Therefore, the focus of this paper is to push beyond these assumptions and develop tools for more general graphs.

\subsection{Our Results and Techniques}
We focus our attention on trees, the simplest graphs outside the reach of prior techniques. In addition to modeling certain types of social networks (including hierarchical ones, or communication networks in which redundancy is expensive), and forming the backbone of many more, trees already present a number of technical challenges whose absence enabled the prior results.  We study preferential attachment trees, which are well-studied graphs with rich structure\footnote{The more general preferential attachment graphs are a popular model of real networks.}. Our main result is the following:

\begin{theorem} Let $G$ be a tree. Then with probability $1-o(1)$\footnote{As $n\to \infty$ the probability converges to $1$, where $n$ is the number of nodes in $G$.}, asynchronous majority dynamics in $G$ stabilizes in a \correct majority if:
\begin{itemize}
\item $G$ is formed according to the preferential attachment model.\footnote{That is, $G$ is created by adding nodes one at a time. When a node is added, it attaches a single edge to a random previous node, selected proportional to its degree.}
\item $G$ is a balanced, $M$-ary tree of any degree.\footnote{That is, $G$ can be rooted at some node $v$. All non-leaf nodes have $M$ children, and all leaves have the same distance to $v$.}
\end{itemize}
\end{theorem}

\paragraph*{Beyond Prior Tools.} In prior work~\cite{FeldmanILW14}, the authors have two key ideas. Without yet getting into full details, one key idea crucially invokes sparsity to claim that most pairs of nodes $u,v$ have distance $d(u,v) = \Omega(\ln n /\ln \ln n)$, which allows them to conclude that after $O(n\ln n/\ln \ln n)$ steps, most nodes are announcing a \correct opinion.

Still, just the fact that the dynamics hit a \correct majority along the way does not imply that the \correct majority will hold thru termination. To wrap up, they crucially invoke expansiveness (building off an argument of~\cite{MosselNT13}) to claim that once there is a \correct majority, it spreads to a \correct consensus with high probability.

Both properties are necessary for prior work, and both properties fail in trees. For instance, the star graph is a tree, but $d(u,v) \leq 2$ for all $u,v$ (precluding their ``majority at $O(n \ln n/\ln \ln n)$'' argument). Additionally, trees are not expansive. In particular, the line graph discussed earlier is a tree which hits a \correct majority at some point (as this tree happens to be sparse), but does not converge to consensus, so there is no hope for an argument like this. However, we believe that the process stabilizes in a correct majority in all trees.

\begin{conjecture*}
Let $G$ be any tree. Then the asynchronous majority dynamics stabilizes in a \correct majority with probability $1 - o(1)$.
\end{conjecture*}

\paragraph*{New Tools.} Our main technical innovation is an approach to reason about majority without going through consensus. Specifically, we show in Sections~\ref{sec:stable} and~\ref{sec:wrapup} for preferential attachment trees, or balanced $M$-ary trees, that with probability $1-o(1)$, a $1-o(1)$ fraction of nodes have \emph{finalized} after $O(n\ln n / \ln \ln n)$ steps. That is, after $O(n\ln n /\ln \ln n)$ steps, most (but not all) of the network has stabilized. The main barrier to extending our results to general trees is Section~\ref{sec:stable}, as we require additional structure on the graphs to prove that the process stabilizes quickly. We postpone further details to Section~\ref{sec:stable}, but just wish to highlight this approach as a fairly significant deviation from prior work.

From here, our task is now reduced to showing that a \correct majority exists w.h.p. after $O(n\ln n / \ln \ln n)$ steps. Our main insight here is that most nodes with $d(u,v) = O(\ln n/\ln \ln n)$ must have some high-degree nodes along the path from $u$ to $v$. We prove that such nodes act like a ``road block,'' causing announcements on either side to be independent with high probability (and all nodes with $d(u,v) = \Omega(\ln n / \ln \ln n)$ can be handled with similar arguments to prior work).


\subsection{Related Work}\label{sec:related}
Information aggregation in social networks is an enormous field, and we will not come close to overviewing it in its entirety. Below, we'll briefly summarize the most related literature, restricting attention to works that consider two states of the world and independent initial beliefs are independently \correct with probability $1/2+\delta$.

\paragraph*{Bayesian Dynamics.} In Bayesian models, agents are fully rational and sequentially perform Bayesian updates to their public opinion based on the public opinions of their neighbors. Seminal works of Banerjee~\cite{Banerjee92} and Bikchandani, Hirshleifer, and Welch~\cite{BikchandaniHW92} first identified the potential of information cascades in this model. Subsequent works consider numerous variations, aiming to understand what assumptions on the underlying network or information structure results in \correct consensus~\cite{SmithS00,BanerjeeF04,AcemogluDLO11}. Many other works studied repeated interactions of Bayesian agents in Social Networks~\cite{GaleK03,RosenbergSV09,KanoriaT13,MullerF13,MosselST14,MosselOT16}. While the high-level goals of these works align with ours, technically they are mostly unrelated as we consider non-Bayesian dynamics. 

\paragraph*{Voter and DeGroot Dynamics.} Prior work also considered other non-Bayesian dynamics. In voter dynamics, individuals update by copying a random neighbor~\cite{Cliffords73,HolleyL75}. Similar dynamics (such as $3$-majority, or $k$-majority) are analyzed from a distributed computing perspective with an emphasis on the rate of convergence to consensus~\cite{BechettiCNPL16, GhaffariP16, GhaffariL18}. In the DeGroot model, individuals announce an opinion in $[0,1]$ (as opposed to $\{0,1\}$), and update by averaging their neighbors~\cite{DeGroot74, GolubJ10}. The biggest difference between these works and ours is that consensus is reached with probability $1$ in these models on any connected graph, which doesn't hold for majority dynamics.

\paragraph*{Majority Dynamics.} The works most related to ours consider majority dynamics. Even synchronous majority dynamics may not result in a consensus (consider again the line graph). These works, like ours, therefore seek to understand what graph structures result in a \correct majority. Mossel, Neeman, and Tamuz study synchronous majority dynamics and prove that a \correct majority arises as long as the underlying graph is sufficiently symmetric, or sufficiently expansive (in the latter case, they prove that the network further reaches consensus)~\cite{MosselNT13}. Feldman et al. study asynchronous majority dynamics and prove that a \correct consensus arises when the underlying graph is sparse and expansive~\cite{FeldmanILW14}. Work of~\cite{TamuzT13} further studies ``retention of information,'' which asks whether \emph{any} recovery procedure (not necessarily a majority vote) at stabilization can recover the ground truth with high probability. In connection to these, our work simply pushes the boundary beyond what classes of graphs are understood in prior work. 

The key difference between the synchronous and asynchronous models is captured by the complete graph. In asynchronous dynamics, a \correct majority occurs only with probability $1/2+\delta$, whereas in synchronous dynamics a \correct consensus occurs with probability $1-\text{exp}(-\Omega(n))$. This is because in step one, every node simply announces their private belief, and in step two everyone updates to the majority, which is $\correct$ with probability $1-\text{exp}(-\Omega(n))$. So while the models bear some similarity, and some tools are indeed transferable (e.g. the expansiveness lemma of~\cite{MosselNT13} used in~\cite{FeldmanILW14}), much of the anlayses will necessarily diverge.

\paragraph*{Preferential Attachment and Balanced $M$-ary Trees.} There is also substantial prior work studying aggregation dynamics in trees. Here, the most related work is~\cite{Howard00, KanoriaM11a}, which studies synchronous majority dynamics in balanced $M$-ary trees. Less related are works which study ``bottom-up'' dynamics in balanced $M$-ary trees~\cite{KanoriaM11b, ZhangCPMH12, TianST18}, $k$-majority dynamics in preferential attachment trees~\cite{AbdullahBF15}, or model cascades themselves as a preferential attachment tree~\cite{GomezKK11}. While these works provide ample motivation for restricting attention to preferential attachment trees, or balanced $M$-ary trees, they bear no technical similarity to ours. 
\section{Model and Preliminaries}\label{sec:prelim}
We consider an undirected tree $G=(V,E)$ with $|V(G)|=n$ and $|E(G)|=m$. We denote by $\deg(v)$ the degree of a node $v\in V(G)$, $N(v)$ to be its neighbors $\{u,(u,v)\in E\}$, and $d(u,v)$ to be the length of the unique path between $u$ and $v$, and let $P(u,v)$ denote the ordered list of vertices on this path (i.e. starting with $u$ and ending with $v$). We'll also denote by $D(G) = \max_{u,v}\{d(u,v)\}$ the diameter of $G$.

Individuals initially have one of two private beliefs, which we'll refer to as \correct (or $1$) and \inc (or $0$). That is, each $v\in V(G)$ receives an independent private signal $X(v)\in\{0,1\}$, and $\Pr[X(v) = 1]=1/2+\delta$, for some constant $0<\delta<1/2$. 

Individuals also have a \emph{publicly announced} opinion (which we will simply refer to as an announcement). We define $C^t(v)\in\{\bot,0,1\}$ to be the public announcement of $v\in V(G)$ at time $t$. Initially, no announcements have been made, \textit{i.e.} $C^0(v)=\bot$ for all $v$. In each subsequent step, a \emph{single node} $v^t$ is chosen uniformly at random from $V(G)$ and updates her announcement (announcements of all other nodes stay the same)\footnote{This makes the process \emph{asynchronous}.}. $v^t$ updates her announcement using \emph{majority dynamics}. That is, if $N^t_1(v)$ denotes the number of $v$'s neighbors with a \correct announcement at time $t$, and $N^t_0(v)$ denotes the number of $v$'s neighbors with an \inc announcement, then:
\[ C^t(v)  =\left\{
    \begin{array}{ll}
      1 & \text{if } N^{t-1}_1(v)>N^{t-1}_0(v), \text{ and }v=v^t,\\
      0 & \text{if }  N^{t-1}_1(v)<N^{t-1}_0(v), \text{ and }v=v^t,\\
      X(v) & \text{if }  N^{t-1}_1(v)=N^{t-1}_0(v), \text{ and }v=v^t,\\
	C^{t-1}(v) & \text{if } v\neq v^t.\\
   \end{array}
\right. \]

Note that we will treat $\delta$ as an absolute constant. Therefore, the only variable taken inside Big-Oh notation is $n$, the number of nodes (and, for instance, when we write $o(1)$ we mean any function of $n$ that approaches $0$ as $n$ approaches $\infty$).

As shown in \cite{FeldmanILW14}, it is easy to see that in any network this process stabilizes with high probability in $O(n^2)$ steps. That is, the network reaches a state where no node will want to change its announcement and thus the process terminates.

\subsection{Concentration Bounds and Tools from Prior Work}
Our work indeed makes use of some tools from prior work to get started, which we state below. The concept of a \emph{critical time}, defined below, is implicit in~\cite{FeldmanILW14}.

\begin{definition}
The \emph{critical time}\footnote{This definition can be naturally extended to any general graph.} from $u$ to $u$, $T(u,u)$, is the first time that node $u$ announces. The critical time from $u$ to $v$, $T(u,v)$, is recursively defined as the first time that $v$ announces \emph{after} the critical time from $u$ to $x$, where $x$ is the neighbor of $v$ in $P(u,v)$. We further denote the \emph{critical chain} from $u$ to $v$ as the ordered list of critical times from $u$ to $x$ for all $x$ on $P(u,v)$.
\end{definition}

The following lemma is a formal statement of ideas from prior work (a proof appears in Appendix~\ref{app:prelim}). To parse it, it will be helpful to think of the process as first drawing a countably infinite sequence $S$ of nodes to announce, which then allows each $C^t(v)$ to be written as a deterministic function of the random variables $\{X(u), u \in V\}$. Lemma~\ref{lem:influence} below states that in fact, for early enough $t$, initial beliefs for only a proper subset of $V$ suffice.

\begin{lemma}[\cite{FeldmanILW14}]\label{lem:influence}
For all $t$, and all $v$, $C^t(v)$ can be expressed as a function of the subset of signals $\{X(u), T(u,v) \leq t\}$.
\end{lemma}

The final theorem we take from prior work is due to Mossel et al., and is used to claim that at minimum the \emph{expected} number of \correct nodes at termination is at least $(1/2+\delta)n$.

\begin{theorem}[\cite{MosselNT13}]\label{thm:boolean}
Let $f$ be an odd, monotone Boolean function. Let $X_1,\ldots,X_n$ be input bits, each sampled i.i.d. from a distribution that is 1 with probability $p\geqslant 1/2$ and 0 otherwise. Then
$\mathbb{E}[f(X_1,\ldots, X_n)] \geqslant p$.
\end{theorem}

Note that, as long as $v$ has announced at least once by $t$, $C^t(v)$ is an odd, monotone, Boolean function in variables $\{X(u), u \in V\}$,\footnote{That is, flipping all $X(v)$ simultaneously to $1-X(v)$ would cause $C^t(v)$ to flip (odd), and changing any subset of initial beliefs from $0$ to $1$ cannot change $C^t(v)$ from $1$ to $0$ (monotone).} and therefore $\Pr[C^t(v) = 1] \geq 1/2+\delta$ for all $v$ and $t \geq T(v,v)$.

Finally, we'll make use of the following concentration bound on $T(u,v)$ repeatedly. Its proof is a simple application of a Chernoff bound and appears in Appendix~\ref{app:prelim}.

\begin{lemma}\label{lem:path} For all $0<\beta<1$:
\begin{itemize}
\item $\Pr[T(u,v) > 8\cdot \max\{\ln(1/\beta),d(u,v)+1\} \cdot n] \leq \beta^2$.
\item $\Pr[T(u,v)<  (d(u,v)+1)\cdot \beta \cdot n] \leq e^{-\beta d(u,v) (1-\beta)^2/3}$.
\item $\Pr[T(u,v)<  (d(u,v)+1)\cdot \beta \cdot n] \leq \left(e\beta\right)^{d(u,v)} =e^{(1+\ln \beta)\cdot d(u,v)}$.
\end{itemize}
\end{lemma}

\section{Key Concepts}\label{sec:concepts}
Before getting into our proofs, we elaborate some key concepts that will be used throughout. In Proposition~\ref{prop:critical} below, we analyze the connection between critical chains and switches in announcements. Intuitively, Proposition~\ref{prop:critical} is claiming that every fresh announcement can cause other nodes to switch a previous announcement along critical chains, but that these are the \emph{only} switches that can occur. 

\begin{definition}
Let $v$ change her announcement at $t$, and her previous announcement be made at $t'>0$. We say that node $u$ is a \emph{cause} of $v$ changing her announcement at $t$ if $C^t(u) = C^t(v)$, and $C^{t'}(u) \neq C^t(v)$. Observe that every such change in announcement has a cause.
\end{definition}

\begin{proposition}\label{prop:critical}
If $C^t(v) \neq C^{t-1}(v)$, then there exists a node $u$ such that:
\begin{itemize}
\item $t = T(u,v)$ (i.e. the influence of $u$ just reaches $v$ at time $t$).
\item $C^{T(u,u)} (u)= C^t(v)$ (i.e. $v$ is updating to match $u$'s initial announcement).
\item Denote $u = x_0,x_1,\ldots,x_{d(u,v)}=v$ the path $P(u,v)$. Then \emph{every} $x_i$, $i >0$, has $C^{T(u,x_i)-1}(x_i) = C^{t-1}(v)$ and $C^{T(u,x_i)}(x_i) = C^t(v)$, and $x_{i-1}$ caused this change (i.e. \emph{every} node along the path from $u$ to $v$ \emph{changed} to match $u$'s initial announcement). 
\end{itemize}
\end{proposition}
\begin{proof}
The proof proceeds by induction on $t$. Consider $t=1$ as a base case. If $C^1(v) \neq C^0(v) = \bot$, then it must be because $v$ announced at time $1$, meaning that $1 = T(v,v)$ as desired.

Now assume that for all $v$ and all $t' < t$ the claim holds, and consider time $t$. If $v$ does not announce at time $t$ then the claim vacuously holds. If $v$ announces at time $t$ but does \emph{not} change their announcement, then again the claim vacuously holds. If $v$ announces at time $t$ for the first time, then $v$ itself is the desired $u$ and the claim holds. The remaining case is if $v$ changes their previous announcement that was made at time $t' < t$ (and $v$ did not announce between $t'$ and $t$).

Let's consider the state of affairs at time $t'$, when $v$ announced some opinion $A$. This means that, at time $t'$, a majority (tie-breaking for $X(v)$) of $v$'s neighbors were announcing $A$. Yet, at time $t$, a majority (tie-breaking for $X(v)$) of $v$'s neighbors were announcing $B=1- A$. Therefore, \emph{some} node adjacent to $v$ must have switched its announcement to $B$ at some $t'' \in (t',t)$, and stays $B$ till time $t$ (and caused the change). Call this node $x$. We now wish to invoke the inductive hypothesis for $x$ at $t''$.

The inductive hypothesis claims there there is some $u$ (maybe $u = x$) such that $u$ made $B$ as its first announcement, and then every node $y$ along the critical chain from $u$ to $x$ switched from $A$ to $B$ at $T(u,y)$ (caused by its predecessor), and that $t'' = T(u,x)$. Let's first consider the case that $v$ is not on the path from $u$ to $x$ (and therefore $x$ is on the path from $u$ to $v$, since they are adjacent). Then as $T(u,x) = t'' \in (t',t)$, and $v$ does not announce in $(t',t)$, we see that $T(u,v) = t$ (immediately by definition of critical times). Moreover, as $P(u,v)$ is simply $P(u,x)$ concatenated with $v$, the inductive hypothesis already guarantees that $u$ announced $B$ at $T(u,u)$, and that every node $y$ on $P(u,v)$ switched from $A$ to $B$ at $T(u,y)$. So the last step is to show that in fact $v$ \emph{must} not be on the path from $u$ to $x$, and then the inductive step will be complete.

Finally, we show that we cannot have $v$ on the path from $u$ to $x$, completing the inductive step. Assume for contradiction that $v$ were on the path from $u$ to $x$. Then as $t'' = T(u,x)$, we would necessarily have $t' \geq T(u,v)$ (immediately from definition of critical times). However, by hypothesis, $C^{t''}(v) = C^{t'}(v) = A$ (the first equality is simply because $v$ does not announce in $(t', t'']$), contradicting the inductive hypothesis that {$v$ caused $x$ to change (because of $u$)}, which would imply instead that $C^{T(u,v)}(v) = C^{t''}(v) = B$. So $v$ cannot be on the path from $u$ to $x$. 
\end{proof}

Below, we make use of Proposition~\ref{prop:critical} to prove that, in any tree\footnote{ Proposition~\ref{prop:critical} and Corollary~\ref{cor:stopping} hold for any general graph. Since we are only interested in trees, we restrict our proofs to just trees for brevity.}, the process terminates quickly (proof in Appendix~\ref{app:concepts}). Note that~\cite{FeldmanILW14} already proves that the process on trees terminates with probability $1-o(1)$ after $O(n^2)$ steps, so Corollary~\ref{cor:stopping} is a strict improvement when the diameter $D(G)= o(n)$.

\begin{corollary}\label{cor:stopping}
Let $T_{\text{stable}}$ denote the last time that a node changes its announcement. Then with probability $1-o(1)$, $T_{\text{stable}} \leq 8\cdot \max\{2\ln(n), D(G)+1\}\cdot n$.
\end{corollary}

Finally, we prove one last proposition which will be used in future sections regarding the probability that a single node announces \correct throughout the process (the proof appears in Appendix~\ref{app:concepts}). Beginning with $v$'s first announcement, because the graph is a tree, prior to $v$'s first announcement all of $v$'s neighbors' announcements are independent. Therefore, it initially seems like we should expect $v$'s initial announcement to be \correct except with probability exponentially small in $\deg(v)$ --- indeed, this would hold if the dynamics were synchronous. However, since the dynamics are \emph{asynchronous}, there's a good chance that $v$ announces before any of its neighbors and simply announces $X(v)$. That is, the probability that $v$'s initial announcement is \inc is at least $\frac{1/2-\delta}{\deg(v)}$, so we cannot hope for such strong guarantees. This observation highlights one (of several) crucial differences between synchronous and asynchronous dynamics. Still, the proposition below shows roughly that the only bad event is $v$ announcing before many of its neighbors. Below for a set $S$, we'll use $C_S^t(v)$ to denote the following modified dynamics: First, set $C^t_S(v) = \inc$ for all $v \in S$, and all $t$. Then, run the asynchronous majority dynamics as normal. In other words, the modified dynamics hard-code an \inc announcement for all nodes in $S$ and otherwise run asynchronous majority dynamics as usual (this extension will be necessary for a later argument).

\begin{definition}
We say that a node $v$ is \emph{safe thru $T$} if $C^t(v) \in \{\bot, 1\}$ for all $t \leq T$. We further say that a node $v$ is \emph{safe thru $T$, even against $S$} if $C_S^t(v) \in \{\bot, 1\}$ for all $t \leq T$.
\end{definition}

\begin{proposition}\label{prop:redalwaysred}
For all $a$, there exist constants\footnote{does not depend on $n$, but may depend on $\delta$.} $b,c$ such that for any $S$ with $|S| = a$, $v\notin S$ and $T\leq n\cdot e^{b\deg(v)}$, $v$ is safe thru $T$, even against $S$, with probability at least $1-c/\deg(v)$.
\end{proposition}

\section{Forming an Initial Majority}\label{sec:majority}
In this section, we prove that in any tree, a \correct majority forms after a near-linear number of steps (but may later fade). The main idea is to show that the announcements of most pairs of nodes are independent with probability $1-o(1)$, and use Chebyshev's inequality to show that the number of \correct announcements therefore concentrates around its expectation. {The independence argument is the crux of the proof.  To show it, we consider three cases depending on the length and degree sequence of the path between a pair of nodes.  If the path is long, then, similar to prior work, there is simply not enough time for the pair to influence each other.  If the path is short, but (some of) the intermediate nodes have high degrees, then these effectively block influence because the announcements of these high-degree nodes is effectively independent of what's happening on the path.  Finally, if the path is short and the intermediate nodes have low degrees, then the pair certainly may influence each other.  However, a counting argument shows there can only be a vanishingly small fraction of such pairs.} The main result of this section is the following:

\begin{theorem}\label{thm:majority}[Majority in trees]
For sufficiently large $n$, any tree on $n$ nodes and any $T \leq \frac{n\ln n}{32 \ln \ln n}$, after $T$ steps, with probability at least $1-O(e^{-\ln n/(24 \ln \ln n)})$, the announcements of at least $(\frac{1}{2}+\frac{\delta}{2}-e^{-T/n})\cdot n$ nodes are \correct. 

Further, for all constants $\gamma > 0$, there exists a constant $\alpha > 0$ such that for sufficiently large $n$, when $T \leq \frac{n \ln^{1-\gamma} n}{\ln \ln n}$, after $T$ steps, with probability at least $1-n^{-\alpha}$, the announcements of at least $(\frac{1}{2}+\frac{\delta}{2}-e^{-T/n})\cdot n$ nodes are \correct. 
\end{theorem}

\noindent First, we analyze the expected number of \correct nodes using Theorem~\ref{thm:boolean} (proof in Appendix~\ref{app:majority}). Note that, the probability that a node $v$ has announced by $T$ is at least $(1-e^{-T/n})$\footnote{The probability that $v$ wasn't chosen in all $T$ rounds is $\left(1 - \left(1 - \frac{1}{n}\right)^T\right)$}.

\begin{lemma}\label{lem:expect}
At any time $T$, the expected number of nodes $v$ with $C^{T}(v) = \correct$ is at least $(1/2+\delta - e^{-T/n})n$.
\end{lemma}

From here, we now need to show that the number of \correct announcements concentrates around its expectation. To this end, we'll show that most pairs of nodes can be written as functions of disjoint initial beliefs, and are therefore independent. Ideas from~\cite{FeldmanILW14} formally show that this suffices:

\begin{definition}
We say that two nodes $u, v$ are $\varepsilon$-disjoint at $t$ if there exist random variables $X_u, X_v$, written as functions of disjoint sets of initial beliefs (and therefore independent), such that $\Pr[C^t(u) = X_u] \geq 1-\varepsilon$ and $\Pr[C^t(v) = X_v] \geq 1-\varepsilon$. 
\end{definition}

\begin{lemma}\label{lem:deviation}[Inspired by~\cite{FeldmanILW14}] Let $\varepsilon^t_{uv}$ be such that $u, v$ are $\varepsilon^t_{uv}$-disjoint at $t$. Then if $\sum_{u, v} \varepsilon^t_{uv} = D$, the number of \correct nodes at time $t$ is within $\delta n/2$ of its expectation with probability $1-\frac{4n+16D}{\delta^2n^2}$.
\end{lemma}

So our remaining task is to upper bound $\sum_{u,v} \varepsilon^t_{uv}$, and this is the point where we diverge from prior work. For a given pair $u,v$, there are three possible cases. Below, case one is most similar to prior work, and cases two/three are fairly distinct.

\paragraph*{Case One: Long Paths.} One possibility is that $d(u,v) \geq f(n)$, for some $f(n)$ to be decided later. The following proposition implies that at any time $T$, the announcements of pairs of nodes at a large enough distance are almost independent. The proof is provided in Appendix~\ref{app:majority}.
\begin{proposition}\label{prop:long}
Let $d(u,v) \geq \max\{kT, f(n)\}$ for $k \geq 4$. Then $\varepsilon^T_{uv} \leq 2e^{-f(n)/24}$, and $\varepsilon^T_{uv} \leq 2e^{(1-\ln k) \cdot f(n)}$.
\end{proposition}

\paragraph*{Case Two: Short Paths A.} Another possibility is that $d(u,v) < f(n)$. Here, there will be two subcases. First, maybe it's the case that $d(u,v)$ is small \emph{and} the product of degrees on the path from $u$ to $v$ is small. In this case, it very well could be that $\varepsilon^t_{uv}$ is large, which is bad. However, we prove that there cannot be many such pairs (and so in total they contribute $o(n^2)$ to the sum). The following lemma shows in fact that even if we remove the restriction that $d(u,v) = O(T)$, there simply cannot be many pairs of nodes such that the product of degrees on $P(u,v)$ is small (proof in Appendix~\ref{app:majority}).

\begin{lemma}\label{lem:shortpath1}
Let $K$ be the set of pairs of nodes $(u,v)$ such that $\prod_{w\in P(u,v)}\deg(w) \leq X$. Then $|K| \leq Xn/2$.
\end{lemma}

\paragraph*{Case Three: Short Paths B.} The final possibility is that $d(u,v)  <f(n)$, and also that $\prod_{w \in P(u,v)} \deg(w)$ is large. In this case, we will prove that with probability $1-o(1)$ there is some block in $P(u,v)$ causing $u$'s and $v$'s announcemnts to be independent (proof in Appendix~\ref{app:majority}). 

\begin{definition} We say that a node $x \in P(u,v)$ \emph{cuts} $u$ from $v$ thru $T$ if some node $y$ in $P(u,x)$ is safe thru $T$ even against $S_y$, where $S_y$ are $y$'s (at most) two neighbors in $P(u,v)$. 
\end{definition}

\begin{lemma}\label{lem:shortpath3}
Let $T$ be any time and $p_x$ be the probability that $x$ cuts $u$ from $v$ thru $T$ \emph{and also} cuts $v$ from $u$ thru $T$. Then $u$ and $v$ are $p_x$-disjoint at $T$.
\end{lemma}

Next, we wish to show that with good probability there is indeed a node on $P(u,v)$ that cuts $u$ from $v$ and also $v$ from $u$ (proofs in Appendix~\ref{app:majority}).

\begin{lemma}\label{lem:shortpath2}
There exist absolute constants $b,d$ such that for any pairs of nodes $u,v\in V$ with $\prod_{w\in P(u,v)}\deg(w) =X$, $d(u,v) \leq \frac{\ln(X)}{d}$, and $T \leq ne^{bX^{1/(4d(u,v))}}$, there exists an $x$ such that with probability $1-X^{-1/8}$, $x$ cuts $u$ from $v$ thru $T$ and also $v$ from $u$ thru $T$.
\end{lemma}

\begin{corollary}\label{cor:caseThree}
There exist absolute constants $b, d$ such that for pairs of nodes any $u,v\in V$ with $\prod_{w\in P(u,v)}\deg(w) =X$, $d(u,v) \leq \frac{\ln(X)}{d}$, and $T \leq ne^{bX^{1/(4d(u,v))})}$, $u$ and $v$ are $X^{-1/8}\text{-disjoint}$ at $T$.
\end{corollary}

Now, we'll put together case one, Lemma~\ref{lem:shortpath3} and Corollary~\ref{cor:caseThree} together to prove Theorem~\ref{thm:majority}, which is mostly a matter of setting parameters straight (and appears in Appendix~\ref{app:majority}).

To conclude, at this point we have proven that a majority takes hold after $n\frac{\ln n}{32\ln \ln n}$ steps for any tree. The remaining work is to prove that it does not disappear.
\section{Stabilizing Quickly}\label{sec:stable}
In this section, we identify properties of a tree which cause it to stabilize quickly. Our main theorem will then follow by proving that both balanced $M$-ary trees and preferential attachment trees have this property. The main idea is to consider nodes that are ``close'' to leaves in the following formal sense:

\begin{definition}
We say that a node $v$ is an $(X,Y)$-leaf in $G$ if there exists a rooting of $G$ such that $v$ has $\leq X$ descendants, and the longest path from $v$ to one of its descendants is at most $Y$. Note that leaves are $(0,0)$-leaves. When we refer to a node's parent, children, or descendants, it will be with respect to this rooting.
\end{definition}

\begin{definition}
We say that a node $v$ is:
\begin{itemize}
\item \emph{finalized} at $T$, if $C^t(v) = C^T(v)$ for all $t \geq  T$.
\item \emph{nearly-finalized} at $T$ with respect to $u$ if there exists a $t' \geq T$ such that $v$ is finalized at $t'$ and for all $t \in (T,t')$ when $v$ announces, it either updates $C^t(v) = C^t(u)$, if $C^t(u) \neq \bot$, or $C^t(v) = C^{t-1}(v)$, if $C^t(u) = \bot$.
\end{itemize}
\end{definition}

Intuitively, a node is finalized if it is done changing its announcement. A node $v$ is nearly-finalized with respect to $u$ if $v$ is not quite finalized, but changes in $u$ are the only reason why $v$ would change its announcement (and moreover, $v$ will copy $u$ every announcement until $v$ finalizes).

The main result of this section is as follows:
\begin{theorem}\label{thm:stable}
Let $v$ be an $(X,Y)$-leaf. Then with probability $1 - Xe^{-T/nY}$, $v$ is nearly-finalized at $T$ with respect to its parent.
\end{theorem}

The main insight for the proof of Theorem~\ref{thm:stable} will be the following lemmas. Below, Lemma~\ref{lem:stablechildren} asserts that once all of $v$'s children are nearly-finalized with respect to $v$, any changes in $v$'s opinion are to copy its parent, and Lemma~\ref{lem:stable} builds off this to claim that we can relate the time until $v$ nearly-finalizes to its critical times. Importantly, Lemma~\ref{lem:stable} does \emph{not} require \emph{all} critical paths to hit $v$, but only those from its descendents.

\begin{lemma}\label{lem:stablechildren}
Let all of $v$'s children be nearly-finalized with respect to $v$ at $T$, and let $u$ be $v$'s parent. Let also $t > t' \geq T$ be two timesteps during which $v$ announced. Then if $C^t(v) \neq C^{t-1}(v)$, we must have $C^t(v) = C^t(u)$.
\end{lemma}

\begin{lemma}\label{lem:stable}
Let $T_v:=\max\{T(x,v), x\text{ is a descendant of $v$}\}$. Then $v$ is nearly-finalized at $T_v$ with respect to its parent.
\end{lemma}

These above lemmas suffice to prove Theorem~\ref{thm:stable}. The proofs of these lemmas and Theorem~\ref{thm:stable} appear in Appendix~\ref{app:stable}.

We will also need the following implications of Theorem~\ref{thm:stable}. Below, Lemma~\ref{lem:flip} will be helpful in proving Corollary~\ref{cor:counting}. Corollary~\ref{cor:counting} lets us claim that while nearly-finalized nodes are not themselves finalized, their existence implies the existence of other finalized nodes. This will be helpful in wrapping up in the following section, since the process only terminates once nodes are finalized. The proofs of Lemma~\ref{lem:flip} and Corollary~\ref{cor:counting} appear in Appendix~\ref{app:stable}.

\begin{lemma}\label{lem:flip}
For any $t>T_v$, if a child of $v$ changes their announcement at $t$, $v$ becomes finalized at $t$.
\end{lemma}

\begin{corollary}\label{cor:counting}
For any $T$, with probability $1-e^{-T/n}$, $v$ has $\lfloor (\deg(v)-1)/2\rfloor$ children who are finalized at $T_v +T$. 

Moreover, if $v$ is an $(X,Y)$-leaf and finalized at $t\geq T_v$, then with probability $1-Xe^{-T/nY}$, all of $v$'s descendants are finalized at $t+T$.
\end{corollary}

\section{Wrapping Up: Preferential Attachment and Balanced \emph{M}-ary Trees}\label{sec:wrapup}
In this section, we show how to make use of Theorem~\ref{thm:stable} to conclude that a $1-o(1)$ fraction of nodes are finalized by $\frac{n\ln n}{32\ln\ln n}$. Proofs for the two cases follow different paths, but both get most of their mileage from the developments in Section~\ref{sec:stable}.

\subsection{Preferential Attachment Trees Stabilize Quickly}
Let's first be clear what we mean by a preferential attachment tree.\footnote{Note that this is the standard definition of preferential attachment used for heuristic arguments, e.g.~\cite{BarabasiA99}. Most prior rigorous work uses a slightly modified definition that produces a forest instead of a tree in order to rigorously analyze (say) the degree distribution~\cite{BollobasRST01,BollobasR04}. As we are only interested in (fairly loose) bounds on the degrees, our results are rigorous in the standard model.}

\begin{definition}[Preferential Attachment Tree]
$n$ nodes arrive sequentially, attaching a single edge to a pre-existing node at random proportional to its degree. Specifically:
\begin{itemize}
\item Let $v_i$ denote the $i^{th}$ node to arrive.
\item Let $\deg_t(v_i)$ denote the degree of node $v_i$ after a total of $t$ nodes have arrived. 
\item There is a special node $v_0$, which only $v_1$ connects to upon arrival, and no future nodes.
\item When $v_{i+1}$ arrives, $v_{i+1}$ attaches a single edge to a previous node, choosing node $v_j$, $j \in [1,i]$, with probability $\frac{\deg_i(v_j)}{2i-1}$.
\end{itemize}
\end{definition}

Our main argument for preferential attachment trees is that most nodes are in a ``good'' subtree, defined below. All subsequent proofs are in Appendix~\ref{app:wrapup}. At a high level the plan is as follows: first, we prove that because most nodes are $(X,Y)$-leaves for small $X,Y$, these nodes quickly become nearly-finalized. Next, we prove that most such nodes are part of a small subtree whose parent is likely to be safe thru the entire process. Therefore, the parent of this subtree is finalized early, and once the subtree becomes nearly-finalized, it finalizes quickly as well.

\begin{definition} Say that a subtree rooted at $v$ is \emph{good} if:
\begin{itemize}
\item $v$ is a $(X,Y)$-leaf, for $X = \ln^{O(1)} n$ and $Y = O(\ln \ln n)$.
\item $v$'s parent has degree at least $\ln^{\Omega(1)} n$.
\end{itemize}
\end{definition}

\begin{proposition}\label{prop:good}
Let the subtree rooted at $v$ be good, and let the diameter of the entire graph be $O(\ln n)$. Then with probability $1-o(1)$, the entire subtree rooted at $v$ is finalized by $n\frac{\ln n}{32 \ln \ln n}$.
\end{proposition}

\begin{proposition}\label{prop:PA}
For a tree built according to the preferential attachment model, the following simultaneously hold with probability $1-o(1)$.
\begin{itemize}
\item $n-o(n)$ nodes are in good subtrees.
\item The diameter of the entire graph is $O(\ln n)$.
\end{itemize}
\end{proposition}

\begin{theorem}\label{thm:PA}
A tree built according to the preferential attachment model stabilizes in a \correct majority with probability $1-o(1)$.
\end{theorem}

The proofs for Proposition~\ref{prop:good}, Proposition~\ref{prop:PA}, and Theorem~\ref{thm:PA} appear in Appendix~\ref{app:wrapup}.

\subsection{Balanced \emph{M}-ary Trees Stablize Quickly}
Let's first be clear what we mean by a balanced $M$-ary tree.

\begin{definition}
We say a tree is a balanced $M$-ary tree if there is a root $v$ such that all non-leaf nodes have exactly $M$ children, and all root-leaf paths have the same length.
\end{definition}

Our plan of attack is as follows (all proofs are in Appendix~\ref{app:wrapup}). First, the case for large $M$ (say, $M > \ln n$) is actually fairly straight-forward as a result of Proposition~\ref{prop:redalwaysred}. This is because every pair of nodes has a high-degree block on their path, meaning that the ``Case Three'' argument used in Section~\ref{sec:majority} actually applies all the way until the process terminates. The $M \leq \ln n$ case is more interesting, and requires the tools developed in Section~\ref{sec:stable}. 

Here, the plan is as follows. Corollary~\ref{cor:counting} roughly lets us claim that all nearly-finalized nodes must have a decent number of finalized children, and moreover that all these finalized children have finalized descendents. Iterating this counting inductively through children, we see that actually most descendents of nearly-finalized nodes of sufficient height must themselves be finalized.

Formally, the approach is to first get a bound on the height for which we can claim that nodes are indeed nearly-finalized with high probability (Corollary~\ref{cor:height}, immediately from Theorem~\ref{thm:stable}). $\ln \ln n$ turns out to be a good choice.

\begin{corollary}\label{cor:height}
Let $v$ be distance $h$ from a leaf. Then $v$ is an $(2M^h, h)$-leaf, and therefore $v$ is nearly-finalized with respect to its parent at $\frac{n\ln n}{64\ln\ln n}$ with probability $1 - 2M^h \cdot e^{-\frac{\ln n}{64h\ln\ln n}} =1 - e^{-\frac{\ln n}{64h\ln\ln n}+h\ln (2M)}$.

In particular, if $h = o(\sqrt{\frac{\ln n}{\ln \ln n \cdot \ln M}})$, then $v$ is nearly-finalized with respect to its parent at $\frac{n\ln n}{64 \ln\ln n}$ with probability $1-o(1)$.
\end{corollary}

\begin{proposition}\label{prop:countingMary}
Let $v$ have height $h = \ln \ln n$ in a balanced $M$-ary tree for $M \leq \ln n$. Then with probability $1-o(1)$, at most $2M^h \cdot (2/3)^{h} = o(M^h)$ descendents of $v$ are not finalized by $\frac{n\ln n}{32 \ln \ln n}$. 
\end{proposition}

\begin{theorem}\label{thm:Mary}
Any $M$-ary tree stablizes in a \correct majority with probability $1-o(1)$.
\end{theorem}

The proofs of Proposition~\ref{prop:countingMary} and Theorem~\ref{thm:Mary} appear in Appendix~\ref{app:wrapup}.
\bibliographystyle{plain}
\bibliography{SNC-bib}

\begin{thebibliography}{10}

\bibitem{AbdullahBF15}
Mohammed~Amin Abdullah, Michel Bode, and Nikolaos Fountoulakis.
\newblock Local majority dynamics on preferential attachment graphs.
\newblock In {\em Proceedings of the 12th International Workshop on Algorithms
  and Models for the Web Graph - Volume 9479}, WAW 2015, pages 95--106, Berlin,
  Heidelberg, 2015. Springer-Verlag.

\bibitem{AcemogluDLO11}
Daron Acemoglu, Munther~A. Dahleh, Ilan Lobel, and Asuman Ozdaglar.
\newblock Bayesian learning in social networks.
\newblock {\em Review of Economic Studies}, 78(4):1201--1236, 2011.

\bibitem{BanerjeeF04}
Abhijit Banerjee and Drew Fudenberg.
\newblock Word-of-mouth learning.
\newblock {\em Games and Economic Behavior}, 46(1):1--22, January 2004.

\bibitem{Banerjee92}
Abhijit~V. Banerjee.
\newblock A simple model of herd behavior.
\newblock {\em The Quarterly Journal of Economics}, 107(3):797--817, 1992.

\bibitem{BarabasiA99}
Albert-L{\'a}szl{\'o} Barab{\'a}si and R{\'e}ka Albert.
\newblock Emergence of scaling in random networks.
\newblock {\em Science}, 286(5439):509--512, 1999.

\bibitem{BechettiCNPL16}
Luca Bechetti, Andrea~E.F. Clementi, Emanuele Natale, Francesco Pasquale, and
  Luca Trevisan.
\newblock Stabilizing consensus with many opinions.
\newblock In {\em ACM Symposium on Discrete Algorithms (SODA)}, 2016.

\bibitem{BikchandaniHW92}
Sushil Bikhchandani, David Hirshleifer, and Ivo Welch.
\newblock A theory of fads, fashion, custom, and cultural change in
  informational cascades.
\newblock {\em Journal of Political Economy}, 100(5):992--1026, October 1992.

\bibitem{BollobasR04}
B{\'{e}}la Bollob{\'{a}}s and Oliver Riordan.
\newblock The diameter of a scale-free random graph.
\newblock {\em Combinatorica}, 24(1):5--34, 2004.

\bibitem{BollobasRST01}
B{\'{e}}la Bollob{\'{a}}s, Oliver Riordan, Joel Spencer, and G{\'{a}}bor~E.
  Tusn{\'{a}}dy.
\newblock The degree sequence of a scale-free random graph process.
\newblock {\em Random Struct. Algorithms}, 18(3):279--290, 2001.

\bibitem{Cliffords73}
Peter Clifford and Aidan Sudbury.
\newblock A model for spatial conflict.
\newblock {\em Biometrika}, 60(3):581--588, 1973.

\bibitem{DeGroot74}
Morris~H. DeGroot.
\newblock Reaching a consensus.
\newblock {\em Review of Economic Studies}, 69(345):118--121, 1974.

\bibitem{DommersHH09}
Sander Dommers, Remco van~der Hofstad, and Gerard Hooghiemstra.
\newblock Diameters in preferential attachment models.
\newblock {\em Journal of Statistical Physics}, 139(1):72--107, Apr 2010.

\bibitem{FeldmanILW14}
Michal Feldman, Nicole Immorlica, Brendan Lucier, and S.~Matthew Weinberg.
\newblock Reaching consensus via non-bayesian asynchronous learning in social
  networks.
\newblock In {\em Approximation, Randomization, and Combinatorial Optimization.
  Algorithms and Techniques, {APPROX/RANDOM} 2014, September 4-6, 2014,
  Barcelona, Spain}, 2014.

\bibitem{GaleK03}
Douglas Gale and Shachar Kariv.
\newblock Bayesian learning in social networks.
\newblock {\em Games and Economic Behavior}, 45(2):329 -- 346, 2003.
\newblock Special Issue in Honor of Robert W. Rosenthal.

\bibitem{GhaffariL18}
Mohsen Ghaffari and Johannes Lengler.
\newblock Nearly-tight analysis for 2-choice and 3-majority consensus dynamics.
\newblock In {\em ACM Symposium on Principles of Distributed Computing (PODC)},
  2018.

\bibitem{GhaffariP16}
Mohsen Ghaffari and Merav Parter.
\newblock A polylogarithmic gossip algorithm for plurality consensus.
\newblock In {\em ACM Symposium on Principles of Distributed Computing (PODC)},
  2016.

\bibitem{GolubJ10}
Benjamin Golub and Matthew~O. Jackson.
\newblock Na\"{i}ve learning in social networks and the wisdom of crowds.
\newblock {\em American Economic Journal: Microeconomics}, 2(1):112--149, 2010.

\bibitem{GomezKK11}
Vicen{\c{c}} G{\'o}mez, Hilbert~J Kappen, and Andreas Kaltenbrunner.
\newblock Modeling the structure and evolution of discussion cascades.
\newblock In {\em Proceedings of the 22nd ACM conference on Hypertext and
  hypermedia}, pages 181--190. ACM, 2011.

\bibitem{HolleyL75}
Richard~A. Holley and Thomas~M. Liggett.
\newblock Ergodic theorems for weakly interacting infinite systems and the
  voter model.
\newblock {\em The Annals of Probability}, 3(4):643--663, 1975.

\bibitem{Howard00}
C~Douglas Howard.
\newblock Zero-temperature ising spin dynamics on the homogeneous tree of
  degree three.
\newblock {\em Journal of applied probability}, 37(3):736--747, 2000.

\bibitem{Janson14}
Svante Janson.
\newblock Tail bounds for sums of geometric and exponential variables.
\newblock {\em CoRR}, abs/1709.08157, 2014.

\bibitem{KanoriaT13}
Y.~{Kanoria} and O.~{Tamuz}.
\newblock Tractable bayesian social learning on trees.
\newblock {\em IEEE Journal on Selected Areas in Communications},
  31(4):756--765, 2013.

\bibitem{KanoriaM11b}
Yashodhan Kanoria and Andrea Montanari.
\newblock Subexponential convergence for information aggregation on regular
  trees.
\newblock In {\em Decision and Control and European Control Conference
  (CDC-ECC), 2011 50th IEEE Conference on}, pages 5317--5322. IEEE, 2011.

\bibitem{KanoriaM11a}
Yashodhan Kanoria, Andrea Montanari, et~al.
\newblock Majority dynamics on trees and the dynamic cavity method.
\newblock {\em The Annals of Applied Probability}, 21(5):1694--1748, 2011.

\bibitem{MosselOT16}
E.~{Mossel}, N.~{Olsman}, and O.~{Tamuz}.
\newblock Efficient bayesian learning in social networks with gaussian
  estimators.
\newblock In {\em 2016 54th Annual Allerton Conference on Communication,
  Control, and Computing (Allerton)}, 2016.

\bibitem{MosselNT13}
Elchanan Mossel, Joe Neeman, and Omer Tamuz.
\newblock Majority dynamics and aggregation of information in social networks.
\newblock In {\em Autonomous Agents and Multi-Agent Systems (AAMAS)}, 2013.

\bibitem{MosselST14}
Elchanan Mossel, Allan Sly, and Omer Tamuz.
\newblock Asymptotic learning on bayesian social networks.
\newblock {\em Probability Theory and Related Fields}, 2014.

\bibitem{MullerF13}
Manuel Mueller-Frank.
\newblock A general framework for rational learning in social networks.
\newblock {\em Theoretical Economics}, 8(1):1--40, 2013.

\bibitem{Pittel94}
Boris Pittel.
\newblock Note on the heights of random recursive trees and random m-ary search
  trees.
\newblock {\em Random Structures \& Algorithms}, 5(2):337--347, 1994.

\bibitem{RosenbergSV09}
Dinah Rosenberg, Eilon Solan, and Nicolas Vieille.
\newblock Informational externalities and emergence of consensus.
\newblock {\em Games and Economic Behavior}, 66(2):979 -- 994, 2009.

\bibitem{SmithS00}
Lones Smith and Peter Sorensen.
\newblock Pathological outcomes of observational learning.
\newblock {\em Econometrica}, 68(2):371--398, March 2000.

\bibitem{TamuzT13}
Omer Tamuz and Ran Tessler.
\newblock Majority dynamics and the retention of information.
\newblock In {\em Working paper}, 2013.

\bibitem{TianST18}
Y.~P. Tian, X.~J. Sun, and O.~Tian.
\newblock Detection performance of the majority dominance rule in $m$ -ary
  relay trees with node and link failures.
\newblock {\em IEEE Transactions on Signal Processing}, 66(6):1469--1482, March
  2018.

\bibitem{ZhangCPMH12}
Zhenliang Zhang, Edwin~KP Chong, Ali Pezeshki, William Moran, and Stephen~D
  Howard.
\newblock Learning in hierarchical social networks.
\newblock {\em IEEE Journal of Selected Topics in Signal Processing},
  7(2):305--317, 2013.

\end{thebibliography}

\appendix
\section{Omitted Proofs From Section~\ref{sec:prelim}}\label{app:prelim}

\begin{proof}[Proof of Lemma~\ref{lem:influence}]
{The proof proceeds by induction on $t$. The claim trivially holds for a base case of $t=0$, as all $C^0(v)$ are deterministically $\bot$. Assume for inductive hypothesis that the claim holds for $t-1$, and now consider $t$. First, consider all $v$ that are not selected to announce. By inductive hypothesis, $C^{t-1}(v)$ can be expressed as a function of $\{X(u), t-1 \geq T(u,v)\}$. As $C^t(v) =C^{t-1}(v)$, and $t \geq t-1$, the inductive step holds for such $v$.

Next, consider the $v$ selected to announce. Clearly, we can write $C^t(v)$ as a function of $\{C^{t-1}(x), (x,v) \in E\} \cup \{X(v)\}$. By inductive hypothesis, each $C^{t-1}(x)$ can be written as a function of $\{X(u), t-1 \geq T(u,x)\}$. Finally, we observe immediately from the definition of critical times that if $T(u,x) \leq t-1$, and $x$ is a neighbor of $v$, and $v$ announces at $t$, then $T(u,v) \leq t$. Therefore, for all $x$ adjacent to $v$, and all $u$ such that $t-1 \geq T(u,x)$, we also have $t \geq T(u,v)$, meaning that we can indeed write $C^t(v)$ as a function of $\{X(u), t \geq T(u,v)\}$.}
\end{proof}

\begin{proof}[Proof of Lemma~\ref{lem:path}]
Let's first analyze the random variable $T(u,v)$. Note that the random variable $T(u,u)$ is just a geometric random variable of rate $1/n$ (because we are waiting for $u$ to be selected to announce). Moreover, once we hit $T(u,u)$, $T(u,x)$ is just an independent geometric random variable of rate $1/n$, where $x$ is the next node on $P(u,v)$. Following the same reasoning, we see that the random variable $T(u,v)$ is the sum of $d(u,v)+1$ i.i.d. geometric random variables of rate $1/n$ (also called a negative binomial distribution with parameters $(d(u,v)+1,1/n)$). 

Moreover, we can couple the event that $T(u,v) > Tn$ (resp. $<Tn$) with the event that the sum of $Tn$ independent Bernoulli's with rate $1/n$ does not exceed $d(u,v)+1$ (resp., exceeds). By the Chernoff bound, both are upper bounded by $e^{-(\frac{d(u,v)+1}{T}-1)^2T/3}$, and plugging in for $T = 8 \cdot \max \{\ln(1/\beta),d(u,v)+1\}$, we get the first statement. Plugging in for $T = \beta\cdot (d(u,v)+1)$, we get the second.

Also by the Chernoff bound (but taking advantage of the fact that perhaps $d(u,v)+1 \gg T$, and using the upper bound $\left(\frac{e}{(1+\delta)}\right)^{(1+\delta)\mu}$), this is upper bounded by $\left(\frac{eT}{d(u,v)+1}\right)^{d(u,v)}$ after substituting for $(1+\delta)\mu = d(u,v)+1$, and $1+\delta = (d(u,v)+1)/T$. The bullet in the lemma comes from further substituting $T = \beta\cdot(d(u,v)+1)$. 
\end{proof}

\section{Omitted Proofs From Section~\ref{sec:concepts}}\label{app:concepts}

\begin{proof}[Proof of Corollary~\ref{cor:stopping}]
{We'll only appeal to Proposition~\ref{prop:critical} and the randomness over the order of announcements, and observe that our bound holds even for worst-case initial private beliefs.  Lemma~\ref{lem:path} guarantees that for a single $(u,v)$ path, $T(u,v) \leq 8 \cdot \max\{2\ln(n), D(G)+1\} \cdot n$ with probability $1-1/n^4$. 

Taking a union bound over all pairs $(u,v)$, we see that with probability $1-o(1)$, each $T(u,v) \leq \max \{2\ln n, D(G)+1\}\cdot n$. Combining with Proposition~\ref{prop:critical} immediately proves the corollary.}

\end{proof}

\begin{proof}[Proof of Proposition~\ref{prop:redalwaysred}]
First, if $\deg(v) \leq a/\delta$, we observe that $c = a/\delta$ is an absolute constant which clearly suffices. So we proceed under the assumption that $\deg(v) > a/\delta$. We begin with a mathematical fact that will be used later in the proof.

\begin{fact}\label{fact:math}
For all $i\geq 0$, there exists an absolute constant $c_i$ such that for all $n \geq 1$, $\sum_{j = 0}^n \frac{e^{-\delta^2 j/12}}{(n-j+1)^i} \leq c_i/n^i$. Moreover, when $i \geq 1$, $c_i \leq \frac{5}{1-e^{-\delta^2/12}}\cdot \left(\frac{24i}{e\delta^2}\right)^i$. When $i \in [0,1]$, $c_i \leq \frac{5}{1-e^{-\delta^2/12}}\cdot \left(\frac{24}{e\delta^2}\right)$. Therefore, there further exists an absolute constant $\alpha = \frac{5}{1-e^{-\delta^2/12}} \cdot \frac{24}{e\delta^2}$ such that $c_i \leq (\alpha i)^i$ for all $i \geq 1$, and $c_i \leq \alpha$ for all $i \in [0,1]$. 
\end{fact}
\begin{proof}
Let's bucket the terms so that bucket $k$ has terms ranging from $n(1-2^{-k/i})$ thru $n(1-2^{(-k-1)/i})$. Observe that the denominator for all terms in bucket $k$ is at least $n^i \cdot 2^{(-k-1)}$. Observe also that the numerator for all terms in bucket $k$ is at most $e^{-\delta^2 n(1-2^{-k/i})/12}$. Therefore, the sum over all terms in bucket $k$ is upper bounded by:

$$\frac{e^{-\delta^2 n (1-2^{-k/i})/12}}{2^{-k-1}n^i} \cdot \sum_{j=0}^\infty e^{-\delta^2j/12} =\frac{e^{-\delta^2 n (1-2^{-k/i})/12}\cdot 2^{k+1}}{(1-e^{-\delta^2/12})\cdot n^i} = \frac{1}{n^i} \cdot \frac{e^{-\delta^2n/12}}{1-e^{-\delta^2/12}} \cdot e^{\frac{\delta^2n2^{-k/i}}{12}}\cdot 2^{k+1} $$

So summing over all buckets (observing that there are only $i \log_2 n$ buckets), we get that the entire sum is upper bounded by:

$$\frac{1}{n^i}\cdot \frac{e^{-\delta^2 n/12}}{1-e^{-\delta^2/12}} \cdot \sum_{k = 0}^{i \log_2 n} e^{\delta^2n 2^{-k/i}/12} \cdot 2^{k+1}.$$

Now let's break this sum into two parts, first looking at $k \leq i$. For such $k$, we'll simply upper bound $2^{-k/i}$ by $1$, and get that the sum of these terms is at most $2^{i+2} \cdot e^{\delta^2n/12}$. For the $c > i$ terms, we'll upper bound $2^{-k/i}$ by $1/2$, and observe that this part of the sum is at most $e^{\delta^2 n /24} \cdot 4n^i$, yielding that the entire sum is upper bounded by:

$$\frac{1}{n^i}\cdot \frac{e^{-\delta^2n/12}}{1-e^{-\delta^2/12}} \cdot \left( 2^{i+2}\cdot e^{\delta^2n/12} + 4n^ie^{\delta^2n/24}\right) = \frac{1}{n^i} \cdot \left(\frac{2^{i+2}}{1-e^{-\delta^2/12}} + \frac{4n^i e^{-\delta^2 n /24}}{1-e^{-\delta^2/12}}\right).$$

So our last task is to figure out how big the right-most numerator can be. To this end, we can take a derivative with respect to $n$ to see that the term is maximized when $n = 24i/\delta^2$, and that the maximum is $4\left(\frac{24i}{e\delta^2}\right)^i$. So putting everything together, and observing that $2^{i+2} \leq \left(\frac{24i}{e\delta^2}\right)^i$ when $i \geq 1$, we get:

$$c_i \leq \frac{5}{1-e^{-\delta^2/12}}\cdot \left(\frac{24i}{e\delta^2}\right)^i.$$

Notice that $c_i \geq \left(\frac{12i}{e\delta^2}\right)^i$ is necessary, as otherwise the conclusion is false when $n = 12i/\delta^2$ simply by examining the last term in the summand.  When $i \leq 1$, observe that indeed $c_i \leq \frac{5}{1-e^{-\delta^2}}\cdot \frac{24}{e\delta^2}$.
\end{proof}

Now we continue with the analysis. Begin by calling a neighbor of $v$ \emph{corrupted} if it lies on the path from $v$ to an element of $S$. Note that there are at most $a=|S|$ corrupted neighbors. Let $m$ denote the number of $v$'s uncorrupted neighbors (so $m \geq \deg(v) - a$).

Now, let's consider the following (further) modified dynamics where each neighbor $w$ of $v$ \emph{ignores} $v$'s announcement when updating their own (i.e. pretends it is $\bot$). Observe in these modified dynamics that essentially the graph has been disconnected by removing $v$. Therefore, at any time $t$, the announcements of $v$'s neighbors at time $t$ are independent (but of course, $w$'s announcement at $t$ may still be correlated with $w$'s own past announcements). 

In these further modified dynamics, we want to look at the probability that a particular announcement of $v$ is incorrect, \emph{given that exactly $k$ of $v$'s (non-corrupted) neighbors have already announced}. For $k \leq a/\delta$ this probability is clearly at most $1$. For $k > a/\delta$, given that all corrupted nodes are always \inc in the worst case, this announcement of $v$ is \correct whenever $>k/2+a/2$ of the $k$ uncorrupted neighbors are. By Theorem~\ref{thm:boolean}, and the fact that the against $S$ modification only affects corrupted nodes (now that the influence through $v$ is removed), each of the $k$ neighbors is \correct with probability at least $1/2+\delta$, independently. So by a direct application of the Chernoff bound, the probability that at most $k/2+a/2 < (1/2+\delta/2)k$ of these neighbors are \correct is at most $e^{-\delta^2 k (1/2)^2/3} \leq e^{-\delta^2k/12}$.  

Now, we wish to define $P_0(k)$ to be the probability that an announcement of $v$ is \inc in these further modified dynamics given that $k$ of $v$'s neighbors have announced (which we just upper bounded above). We further define $P_1(k)$ to be the probability that an announcement, \emph{or} $v$'s subsequent announcement are \inc, given that $k$ of $v$'s neighbors have announced before the first one. Note that here we are taking the randomness in the announcements of $v$'s neighbors \emph{and} the randomness in more of $v$'s uncorrupted neighbors potentially announcing before the subsequent announcement. Generally, we will define $P_i(k)$ to be the probability that an announcement of $v$ \emph{or any of} $v$'s subsequent $i$ announcements are \inc given that $k$ of $v$'s neighbors have announced before the first one. Note that for any $j \leq m-k$, the probability that exactly $k+j$ of $v$'s neighbors announce before the second announcement, conditioned on exactly $k$ announcing before the first, is exactly $\frac{1}{m-k+1}$. This is simply because the ordering of nodes by their first announcement \emph{after} $v$'s initial announcement is uniformly at random, and whoever winds up before $v$ in this ordering has also announced by $v$'s second announcement. This allows us to conclude the following upper bounds on $P_i(k)$, using the previous paragraph and union bounds. 

\begin{align*}
P_0(k) &\leq 1, \forall k.\\
P_0(k) &\leq e^{-\delta^2k/12}, \forall k \geq a/\delta.\\
P_i(k) &\leq 1, \forall i, k.\\
P_i(k) &\leq e^{-\delta^2k/12} + \sum_{j = 0}^{m-k} \frac{ P_{i-1}(k+j)}{m-k+1}, \forall k \geq a/\delta.\\
\end{align*}

Now, we wish to unwind the recursion to get a true upper bound on $P_i(k)$, making use of Fact~\ref{fact:math}. 

\begin{lemma}\label{lem:recursion}
Let $c_i$ be such that for all $n \geq 1$, $\sum_{j = 0}^n \frac{e^{-\delta^2 j/12}}{(n-j+1)^i} \leq c_i/n^i$. Then for $k \geq a/\delta$:
$$P_i(k) \leq e^{-\delta^2k/12} \cdot \left[1+\sum_{z=1}^i \frac{\prod_{\ell=0}^{z-1} c_\ell}{(m-k+1)^z}\right]$$
\end{lemma}
\begin{proof}
The proof proceeds by induction on $i$. Observe that this holds for $i = 0$ as a base case immediately as $P_0(k) \leq e^{-\delta^2k/12}$. Now assume that this holds for $i$, and consider $P_{i+1}(k)$. Then we have:

\begin{align*}
P_{i+1}(k) &\leq e^{-\delta^2k/12} + \sum_{j = 0}^{m-k} \frac{ P_{i}(k+j)}{m-k+1}.\\
&\leq e^{-\delta^2k/12} + \sum_{j = 0}^{m-k} \frac{e^{-\delta^2(k+j)/12} \cdot \left[1+\sum_{z=1}^i\frac{\prod_{\ell=0}^{z-1} c_\ell}{(m-k-j+1)^z}\right]}{m-k+1}\\
&\leq e^{-\delta^2 k/12}+\frac{e^{-\delta^2k/12}}{m-k+1} \cdot \left[\sum_{j=0}^{m-k} e^{-\frac{\delta^2j}{12}} \cdot \left[1+\sum_{z=1}^i \frac{\prod_{\ell=0}^{z-1}c_\ell}{(m-k-j+1)^z}\right]\right]\\
&\leq e^{-\delta^2 k/12}+\frac{e^{-\delta^2k/12}}{m-k+1} \cdot \left[\sum_{j=0}^{m-k} e^{-\frac{\delta^2j}{12}}+ \sum_{z=1}^i \left(\prod_{\ell=0}^{z-1} c_\ell\right)\cdot \sum_{j=0}^{m-k} \frac{e^{-\delta^2j/12}}{(m-k-j+1)^z}\right]\\
&\leq e^{-\delta^2 k/12}+\frac{e^{-\delta^2k/12}}{m-k+1} \cdot \left[c_0+ \sum_{z=1}^i \left(\prod_{\ell=0}^{z-1} c_\ell \right)\cdot\frac{ c_z}{(m-k+1)^z} \right]\\
&\leq e^{-\delta^2 k/12}\cdot \left[1+\sum_{z=0}^i \frac{\prod_{\ell=0}^z c_\ell}{(m-k+1)^{z+1}}\right]\\
&= e^{-\delta^2 k/12}\cdot \left[1+\sum_{z=1}^{i+1} \frac{\prod_{\ell=0}^{z-1} c_\ell}{(m-k+1)^{z}}\right].\\
\end{align*}

\end{proof}

\begin{corollary}\label{cor:recursion} 
Let $c_i$ be such that for all $n \geq 1$, $\sum_{j = 0}^n \frac{e^{-\delta^2 j/12}}{(n-j+1)^i} \leq c_i/n^i$. Then in the further modified dynamics, the probability that \emph{any} of $v$'s first $i$ announcements are \inc is at most:

$$\frac{a}{\delta(m+1)} + \sum_{z=0}^{i-1} \frac{\prod_{\ell = 0}^z c_\ell}{(m+1)^{z+1}}$$
\end{corollary}

\begin{proof}
Observe that the probability that $k$ of $v$'s uncorrupted neighbors announce before $v$'s first announcement is exactly $1/(m+1)$. Therefore, we get that the probability that any of $v$'s first $i$ announcements are \inc is at most:

\begin{align*}
\frac{1}{m+1}\cdot\sum_{k=0}^m &P_{i-1}(k)\\
 &\leq \frac{1}{m+1} \cdot \left[\sum_{k=0}^{a/\delta-1} 1 + \sum_{k=a/\delta}^m e^{-\delta^2k/12} \cdot \left[1+\sum_{z=1}^{i-1}\frac{\prod_{\ell=0}^{z-1} c_\ell}{(m-k+1)^z}\right]\right]\\
 &\leq \frac{a}{\delta(m+1)} + \frac{1}{m+1}\cdot \left[\sum_{k=0}^m e^{-\delta^2k/12} \cdot \left[1+\sum_{z=1}^{i-1}\frac{\prod_{\ell=0}^{z-1} c_\ell}{(m-k+1)^z}\right]\right]\\
 &\leq \frac{a}{\delta(m+1)} + \frac{1}{m+1}\cdot \left[\sum_{k=0}^m e^{-\delta^2k/12} + \sum_{z=1}^{i-1}\left(\prod_{\ell=0}^{z-1} c_\ell\right) \cdot \sum_{k=0}^m \frac{e^{-\delta^2k/12}}{(m-k+1)^z}\right]\\
&\leq \frac{a}{\delta(m+1)} + \sum_{z=0}^{i-1} \frac{\prod_{\ell=0}^z c_\ell}{(m+1)^{z+1}}
\end{align*}
\end{proof}

Now, we want to combine Corollary~\ref{cor:recursion} with Fact~\ref{fact:math} to get a more transparent bound:

\begin{corollary}\label{cor:nearfinal}
In the further modified dynamics, the probability that \emph{any} of $v$'s first $\frac{\ln(m)}{4\ln \ln(m)}$ announcements are \inc is $O(1/m)$. To be extra clear, there exist two absolute constants $\beta, \gamma$ such that for all $m \geq \beta$, the probability that any of $v$'s first $\frac{\ln (m)}{4\ln \ln (m)}$ announcements are \inc is at most $\gamma/m$.
\end{corollary}
\begin{proof}
We simply plug in the upper bound we have on $c_i$ from Fact~\ref{fact:math}:

\begin{align*}
\prod_{\ell=0}^z c_\ell \leq \alpha\cdot \prod_{\ell=1}^z (\alpha z)^\ell \leq  (\alpha z)^{z^2}\\
\Rightarrow \frac{\prod_{\ell=0}^z c_\ell}{(m+1)^{z+1}} \leq e^{z^2 \ln (\alpha z) - (z+1) \ln (m+1)}.
\end{align*}

Now, observe that when $z = 0$, the RHS is clearly $O(1/m)$. We proceed to handle the the $z > 0$ case, making use of the fact below.

\begin{fact} For sufficiently large $m$, when $z \in [1,\frac{\ln(m)}{4\ln \ln (m)}]$, $z^2 \ln (\alpha z) - (z+1) \ln (m+1) \leq -1.5\ln (m+1)$.
\end{fact}
\begin{proof}
To see this, we first rewrite the claim as $z^2 \ln (\alpha z) \leq (z-0.5) \ln(m+1)$. Next, we take a derivative of both sides with respect to $z$. The derivative of the left-hand side is $2z\ln(\alpha z) + z$. The derivative of the right-hand side is $\ln(m+1)$. In particular, this means that the derivative of the right-hand side exceeds that of the left-hand side whenever $2z \ln(\alpha z) + z \leq \ln(m+1)$. 

For sufficiently large $m$, $z \leq \frac{\ln(m)}{4\ln \ln(m)}$ satisfies this inequality. Next, observe that as a result of this, the inequality is certainly satisfied for all $z \in [1,\frac{\ln(m)}{4\ln \ln(m)}]$ as long as it is satisfied at $z=1$. Observe also that for sufficiently large $m$, the inequality is satisfied at $z = 1$, so there is a sufficiently large $m$ such that both hold.
\end{proof}
So for sufficiently large $m$, and $i \leq \frac{\ln(m)}{4\ln\ln(m)}$, we now have that $\sum_{z = 0}^{i-1} \frac{\prod_{\ell=0}^z c_\ell}{(m+1)^{z+1}} \leq c_0/m+ i/m^{1.5} = O(1/m)$. Adding back in the $\frac{a}{\delta(m+1)}$ term, we get that the entire term is $O(1/m)$. 

\end{proof}

At this point, we have claimed that $v$'s first $\frac{\ln (m)}{4 \ln \ln (m)}$ announcements are correct except with probability $O(1/m)$. We now consider subsequent announcements in two steps. 

\begin{lemma}\label{lem:announcements}
The probability that fewer than $\frac{m}{3e\ln^4(m)}$ of $v$'s uncorrupted neighbors have announced by $v$'s $(\frac{\ln(m)}{4\ln\ln(m)})^{th}$ announcement is at most $1/m$.
\end{lemma}
\begin{proof}
Consider the first $\frac{m}{2e\ln^4(m)}$ timesteps where either $v$ announces, or one of $v$'s neighbors announces for the first time. Note that once these steps are fixed, each such announcement is equally likely to be $v$ or a new announcement of its neighbor. As there are only $\frac{m}{2e\ln^4(m)}$ announcements  total, there are always at least $m/2$ unannounced neighbors, so the probability of picking $v$ is at most $2/m$ for each pick. 

So the probability that $v$ is picked at least $k = \frac{\ln(m)}{4\ln\ln(m)}$ times is at most $\binom{m/(2e\ln^4(m))}{k} \cdot (2/m)^k \leq \frac{m^k\cdot 2^k}{k!\cdot (2e\ln^4(m))^k m^k} \leq \frac{2^k e^k}{k^k e^k2^k\ln^{4k}(m)}$ by Stirling's approximation. Finally, observe that:
$$\frac{1}{k^k (\ln(m))^{4k}} \leq \frac{1}{(\ln(m))^{4k}} = e^{-\frac{\ln(m)}{4\ln\ln(m)} \cdot 4 \cdot \ln \ln (m)} = 1/m.$$

Finally, observe that indeed if at most $\frac{\ln(m)}{4\ln\ln(m)}$ of these announcements are by $v$, then $> \frac{m}{3e \ln^4(m)}$ announcements remain that must be the first announcement from a neighbor of $v$.
\end{proof}

We'll now claim that Lemma~\ref{lem:announcements} is good enough to carry us through the next $m$ announcements. 

\begin{corollary}\label{cor:announcements}
The probability that any of $v$'s announcements from its $\frac{\ln(m)}{4\ln\ln (m)}$ thru its $m^{th}$ are incorrect is $O(1/m)$. 
\end{corollary}
\begin{proof}
First, we just showed that the probability that fewer than $\frac{m}{3e\ln^4(m)}$ of $v$'s neighbors have announced is at most $1/m$. So we'll proceed under the condition that at least this many neighbors have announced. From here, the probability that any given announcement of $v$ is incorrect is at most $e^{-\delta^2 m / (24e\ln^4(m))}$. Taking a union bound over all $m$ announcements results in a total failure probability of $me^{-\delta^2 m/(24e \ln^4(m))} = O(1/m)$. 
\end{proof}

To wrap up this part of the analysis, we now consider any announcements after the first $m$. First, we claim that the probability that \emph{any} of $v$'s neighbors haven't announced by here is $O(1/m)$. 

\begin{lemma}\label{lem:announcementslast}
The probability that any of $v$'s neighbors haven't announced by $v$'s $m^{th}$ announcement is at most $m/2^m = O(1/m)$.
\end{lemma}
\begin{proof}
For a given node $w$, the probability that it doesn't announce by $v$'s $m^{th}$ announcement is exactly $2^{-m}$. To see this, consider the first $m$ times that either $v$ or $w$ announce. In order for $v$ to announce $m$ times before $w$ announces once, each of these announcements must be $v$ and not $w$. But each announcement is equally likely to be $v$ or $w$, independently. So the probability that this happens is $2^{-m}$. The lemma follows from a union bound.
\end{proof}

\begin{corollary}\label{cor:announcementslast}
The probability that any of $v$'s announcements from its $m^{th}$ through it's $L^{th}$ are incorrect is $\leq Le^{-\delta^2 m/12} + O(1/m)$.
\end{corollary}
\begin{proof}
First, we just showed that the probability that any of $v$'s neighbors haven't announced is $O(1/m)$, so we'll proceed under the condition that all neighbors have announced. From here, the corollary is a simple union bound: each announcement has probability at most $e^{-\delta^2m/12}$ of being incorrect. 
\end{proof}

Putting Corollaries~\ref{cor:nearfinal},~\ref{cor:announcements}, and~\ref{cor:announcementslast} together, we get the following:

\begin{corollary}\label{cor:almostthere!}
The probability that any of $v$'s first $L$ announcements are incorrect is $\leq Le^{-\delta^2 m/12} + O(1/m)$.
\end{corollary}

Next, we want to claim that $v$ is unlikely to announce more than $L = e^{\delta^2m/13}$ times in the first $T$ steps, when $T \leq n\cdot e^{b m}$ for some constant $b$, and we'll be done with the analysis of the further modified dynamics. Observe that the expected number of announcements for $v$ in the first $T = n \cdot e^{\delta^2m/13}/2$ steps is exactly $e^{\delta^2 m/13}/2$. By Chernoff bound, the probability that $v$ instead announces at least $e^{\delta^2 m/13}$ times is at most $e^{-e^{\delta^2 m/13}/6} = O(1/m)$. So all together, we are claiming that:

\begin{enumerate}
\item In the further modified dynamics, the first $e^{\delta^2m/13}$ announcements are \correct except with probability $O(1/m)$. 
\item Except with probability $O(1/m)$, $v$ announces at most $e^{\delta^2 m/13}$ times in the first $T = n \cdot e^{\delta^2 m /13}/2$ steps. 
\item So together, in the further modified dynamics, except with probability $O(1/m)$, $v$ is safe thru $T$. 
\end{enumerate}

Finally, we need to reason about the original ``against $S$'' dynamics. We wish to claim that for all orders of announcements, and all initial beliefs, if the further modified dynamics ever result in an \inc majority during one of $v$'s first $\ell$ announcements, then so do the original against $S$ dynamics. To see this, assume for contradiction that $v$'s $i$th announcement was the first \inc announcement in the original against $S$ dynamics, while all of $v$'s first $i$ ``anouncements'' were \correct in the further modified dynamics. Then it must be the case that one of $v$'s neighbors announce \inc in the original dynamics, but \correct in the modified dynamics, \emph{despite $v$ announcing \correct during each of its first $i-1$ announcements}. This is a contradiction, as the inclusion of $v$'s \correct announcemnets cannot cause $w$'s announcement to switch from \correct to \inc. 

\end{proof}

\section{Omitted Proofs From Section~\ref{sec:majority}}\label{app:majority}

\begin{proof}[Proof of Lemma~\ref{lem:expect}]
{Let \textsc{Announced} be the event that a given node $v$ announces by time $T$.  For such a node $v$, we apply Theorem~\ref{thm:boolean}, together with the observation that $C^{T}(v)$ is an odd monotone boolean function of the random variables $\{X(u),u\in V\}$, to conclude that $v$ is \correct with probability at least $1/2+\delta$ (\emph{not} independently of other nodes). Then, we observe that $v$ announces by $T$ with probability at least $1-(1-1/n)^{T/n}\geq 1-e^{-T/n}$ (because each of $T$ announcements are made by a uniformly random one of the $n$ nodes). Therefore, the probability a given node is \correct is $$\Pr[\correct|\textsc{Announced}]\Pr[\textsc{Announced}]\geq (1/2+\delta)(1-e^{-T/n})\geq1/2+\delta-e^{-T/n}$$ 
The lemma follows by linearity of expectation.}
\end{proof}

\begin{proof}[Proof of Lemma~\ref{lem:deviation}]
We first wish to compute the variance of $\sum_v C^t(v)$, which is exactly $\sum_{u,v} \mathbb{E}[C^t(v) \cdot C^t(u)] - \mathbb{E}[C^t(v)] \mathbb{E}[C^t(u)]$. As each $C^t(v) \in \{0,1\}$, and there are only $n$ nodes, terms coming from $u = v$ contribute at most $n$. For pairs of distinct vertices, note that with probability $(1-\varepsilon^t_{uv})^2\geq 1-2\varepsilon^t_{uv}$, both $C^t(v) = X_v$ and $C^t(u) = X_u$. With the remaining $2\varepsilon^t_{uv}$ probability, $C^t(v)C^t(u) \leq 1$. So we get that $\mathbb{E}[C^t(v) \cdot C^t(u)] \leq \mathbb{E}[X_uX_v] + 2\varepsilon^t_{uv}$. Similarly, we see that $\mathbb{E}[C^t(u)] \geq \mathbb{E}[X_u] - \varepsilon^t_{uv}$. This is because in the worst case, whenever $C^t(u) \neq X_u$ is it when $X_u = 1$ and $C^t(u) = 0$, and this happens with probability at most $\varepsilon^t_{uv}$. So we get (again using that $C^t(u), C^t(v) \leq 1$ always):

$$\mathbb{E}[C^t(v) \cdot C^t(u)] - \mathbb{E}[C^t(v)] \mathbb{E}[C^t(u)] \leq \mathbb{E}[X_v X_u] + 2\varepsilon^t_{uv} - \mathbb{E}[X_u]\mathbb{E}[X_v] + 2\varepsilon^t_{uv} \leq 4\varepsilon^t_{uv}.$$

The final inequality comes from the fact that $X_u$ and $X_v$ are independent. So if $\sum_{u,v} \varepsilon^t_{uv} \leq D$, then the total variance of $\sum_v C^t(v) \leq n + 4D$. Now, we simply plug into Chebyshev's inequality, and observe that the probability that a random variable deviates by more than $\gamma$ is at most its variance over $\gamma^2$, yielding the lemma.
\end{proof}

\begin{proof}[Proof of Proposition~\ref{prop:long}]
Let's look at the midpoint $x$, which has distance $d(u,v)/2$ to both $u$ and $v$. Note that if{ $T(x,u)>T$ and $T(x,v)>T$}, then Lemma~\ref{lem:influence} implies that $C^{T}(u)$ and $C^{T}(v)$ can be written as functions of disjoint initial beliefs. More formally, let $X_u = C^{T}(u)$ if $T(x,u) > T$, and \inc otherwise. Similarly, let $X_v = C^{T}(v)$ if $T(x,v) > T$, and \inc otherwise. Then $X_v, X_u$ can be written as functions of disjoint sets of signals, and $C^{T}(v) = X_v$ whenever $T(x,v) > T$. A direct application of Lemma~\ref{lem:path} with $y = 1/4$ yields that the probability that $T(x,u) <T$ is at most $e^{-f(n)/24}$. A union bound (over two events) then lets us conclude by Lemma~\ref{lem:influence} that $\varepsilon^T_{uv} \leq 2e^{-f(n)/24}$.

{The improved statement in the case of large $k$ is also a direct application of the third bullet of Lemma~\ref{lem:path}. Here, we simply plug in $y = 1/k$ to the third bullet.}
\end{proof}

\begin{proof}[Proof of Lemma~\ref{lem:shortpath1}]
For any $u\in V$ and $x>0$, define
\begin{align*}
    K_\ell(u,x)=\left\{v \in V(G) : \prod_{w\in P(u,v)}\deg(w) < x \text{ and } 1 \leqslant d(u,v) \leqslant \ell\right\}.
\end{align*}

Fix a vertex $u$ and consider rooting the tree at $u$. We claim that for all $\ell \in [n]$ and $x>0$, we have $\left\lvert K_\ell(u,x)\right\rvert \leq x$, with equality only if all of $u$'s neighbors are in $K_\ell(u,x)$. If we can prove this, then we conclude that each $u$ participates in at most $X$ pairs in $K$, providing the bound (by summing over all $u$ and dividing by $2$ for double counting). We prove this by induction on the path length $\ell$.

\paragraph*{Base case.} Suppose $\ell=1$. The vertices at distance exactly 1 from $v$ are its neighbors. Fix any $x>0$. We have two possibilities:
\begin{itemize}
    \item $\deg(u)> x$: Note that for all $v\in V(G)$, $\prod_{w\in P(u,v)} \deg(w) \geq \deg(u) \geq x$, implying $|K_1(u,x)|=0\leq x$.
    \item $\deg(u) \leq x$: In this case, $K_1(u,x)\subseteq N(u)$ implies  $|K_1(u,x)| \leq \deg(u) \leq x$. Here, notice that equality is possible \emph{only if} all of $u$'s neighbors are in $K_1(u,x)$.
\end{itemize}

\paragraph*{Inductive hypothesis.} Assume for all nodes $u\in V(G)$ and $x > 0$, we have $|K_{\ell-1}(u,x)|\leq x$, with equality only if all of $u$'s neighbors are in $K_{\ell-1}(u,x)$. 

\paragraph*{Inductive step.} Fix any node $u$. For each $u'\in N(u)$, let us consider the remaining subset of $K_\ell(u,x)$ rooted at $u'$. We claim that the following inequality holds for all $x>0$ {and $\ell\geq2$. Below, let $N^1(u)$ denote the neighbors $u'$ of $u$ that are in $K_\ell(u, x)$ \emph{and} $u \in K_{\ell-1}(u', \frac{x}{\deg(u)})$. Let $N^2(u)$ denote the neighbors $u'$ of $u$ that are in $K_\ell(u,x)$ \emph{but} $u \notin K_{\ell-1}(u', \frac{x}{\deg(u)})$.
\begin{align*}
    |K_{\ell}(u,x)| \leqslant \sum_{u'\in N^1(u)} \left\lvert K_{\ell-1}\left(u',\frac{x}{\deg(u)}\right)\right\rvert + \sum_{u'\in N^2(u)} 1+\left\lvert K_{\ell-1}\left(u',\frac{x}{\deg(u)}\right)\right\rvert.
\end{align*}
To see this, observe that for any node $v$ satisfying $d(u,v)\in\{1,2,\ldots,\ell\}$, there must be a neighbor $u'\in N(u)\cap K_\ell(u,x)$ such that $d(u',v)\in\{0,1,\ldots,\ell-1\}$. Furthermore, for any path $P(u,v)$ satisfying $\prod_{w\in P(u,v)} \deg(w) \leq x$, removing $u$ from $P(u,v)$ gives a new path $P(u',v)$, where $u'$ is a neighbor of $u$, and $\prod_{w\in P(u',v)} \deg(w) \leq \frac{x}{\deg(u)}$.

The right-hand-side is both over-counting and under-counting the size of $K_\ell(u,x)$. It is over-counting in the following two senses.
\begin{itemize}
    \item It may be the case that $K_{\ell-1}\left(u',\frac{x}{\deg(u)}\right) \cap K_{\ell-1}\left(u'',\frac{x}{\deg(u)}\right) \not= \emptyset$ for two neighbors $u',u''\in N(u)$, so some nodes can be counted several times.
    \item It also may be the case that $u \in K_{\ell-1}\left(u',\frac{x}{\deg(u)}\right)$ for any $u'\in N(u)$. In this case, therefore, the node $u$ --- which satisfies $d(u,u)=0$ and is thus not in $K_\ell(u,x)$ --- is counted.
\end{itemize}
The under-counting is due to the fact that the definition of $K_\ell$ excludes the root of the tree. In other words, for each $u'\in N(u)$, $u'\not\in K_{\ell-1}\left(u',\frac{u}{\deg(u)}\right)$. Therefore, perhaps none of the neighbors of $u$ are counted in the right-hand-side. However, this cancels the second case of over-counting listed above: for all $u'$, either $u\in K_{\ell-1}\left(u',\frac{x}{\deg(u)}\right)$, in which case $u$ is overcounted. Or $u \notin K_{\ell-1}\left(u',\frac{x}{\deg(u)}\right)$, in which case the explicit $+1$ cancels the overcounting.

The proof of the claim follows from applying the inductive hypothesis to each term in the sum on the right-hand-side, and recalling that the inequality is strict for all terms in the right-most sum, as their neighbor $u$ is not in $K_{\ell-1}(u', \frac{x}{\deg(u)})$.
\begin{align*}
     |K_{\ell}(u,x)| &\leqslant \sum_{u'\in N^1(u)} \left\lvert K_{\ell-1}\left(u',\frac{x}{\deg(u)}\right)\right\rvert + \sum_{u'\in N^2(u)} 1+\left\lvert K_{\ell-1}\left(u',\frac{x}{\deg(u)}\right)\right\rvert\\
    &\leq \sum_{u' \in N(u) \cap K_\ell(u,x)} \frac{x}{\deg(u)} \leq x.
\end{align*}}
Finally, it is easy to see that the last inequality holds with equality only if every neighbor of $u$ is in $K_\ell(u,x)$.
\end{proof}

\begin{proof}[Proof of Lemma~\ref{lem:shortpath3}]
First, define an auxiliary process that replaces $C^t(x)$ with $\correct$ for all $t \geq T(x,x)$ (and keeps it as $\bot$ before $T(x,x)$), and otherwise runs asynchronous majority dynamics. Under this modified dynamics, it's clear that all announcements of $u$ and $v$ are independent. Announcements under these modified dynamics will serve as the $X_u, X_v$ witnessing $p_x$-disjointness. Our plan is to show that $C^T(v) = X_v$ whenever $x$ cuts $v$ from $u$ thru $T$ (and $C^T(u) = X_u$ whenever $x$ cuts $u$ from $v$ thru $T$). 

Next, consider that when $x$ cuts $u$ from $v$ thru $T$, there is some node $y$ on $P(u,x)$ that is safe thru $T$ even against $S_y$. If $y$ happens to be $x$ itself, then $C^t(u)$ (in the true asynchronous majority dynamics) is accurately computed by the above modified process for all $t$, and therefore the two random variables are equal (ditto for $C^T(v)$ and $X_v$).

If $y$ happens to not be $x$, this means that $C^t(y)$ is known to be \correct (or $\bot$, before its first announcement) \emph{without inspecting $C^{t-1}(z)$ for $z \in S_y$}. Therefore, $C^{t}(y)$ can be written as a function of initial beliefs only of nodes it can reach without going through $S_y$, for all $t \leq T$. And importantly, $C^t(y)$ can be determined at all timesteps without ever inspecting $C^{t'}(z)$, for $z \in S_y$ and any $t'$, which means that $C^t(y)$ can be determined at all timesteps without ever inspecting $C^{t'}(x)$ for any timestep $t'$. So $C^t(y)$ is computed correctly in the auxiliary process because it never looks at $C^{t'}(x)$ anyway (for all $t \leq T$). 

As $y$ is on the path from $x$ to $u$, this further implies that $C^t(u)$ can be computed without ever inspecting $C^{t'}(x)$ for any $t'$, and therefore the auxiliary process also correctly computes $C^t(u)$, for all $t \leq T$. So in conclusion, we have shown that the random variable $X_u$ is equal to $C^T(u)$ whenever $x$ cuts $u$ from $v$. Similarly we have that $X_v$ is $C^T(v)$ whenever $x$ cuts $v$ from $u$. The lemma immediately follows from this claim.
\end{proof}

\begin{proof}[Proof of Lemma~\ref{lem:shortpath2}]
{First, observe that whether or not a node $w$ is safe thru $T$, even against $S_w$, are \emph{independent events}. This is because whether or not $w$ is safe thru $T$, even against $S_w$ can be determined as a function of the signals from nodes in $A_w$, the nodes that can be reached from $w$ without going through $S_w$, and the sets $A_w, w \in P(u,v)$ are disjoint. 

Now, observe that there certainly exists a node $x$ such that $\prod_{w \in P(u,x)}\deg(w) \geq \sqrt{X}$, and $\prod_{w \in P(x,v)} \deg(w) \geq \sqrt{X}$ (this is simply because the product of both terms is at least $X$). This will be our desired $x$. Let's further restrict attention to $P^*(u,x) = \{w \in P(u,x), \deg(w) \geq X^{1/(4|P(u,x)|)}\}$. It's further clear that we still have $\prod_{w \in P^*(u,x)} \deg(w) \geq X^{1/4}$. Now, let's analyze the probability that \emph{no} node $y$ on $P^*(u,x)$ is safe thru $T$, even against $S_y$. By independence, and using Proposition~\ref{prop:redalwaysred}, this is just $\prod_{w \in P^*(u,x)} c/\deg(w)$ for some constant $c$ (note that the application of Proposition~\ref{prop:redalwaysred} is valid by our assumptions on $T$ -- we will simply take the same $b$ promised for $a = 2$). 

We already know that $\prod_{w \in P^*(u,x)} \deg(w) \geq X^{1/4}$, so we get that the total probability that there exists \emph{no} node on $P^*(u,x)$ (and hence, $P(u,x)$) that is safe thru $T$, even against $S_y$ is at most $c^{|P^*(u,x)|}/X^{1/4}$. So now we just wish to see how to bound this. Note that $|P^*(u,x)|\leq d(u,v)$, so this term is upper bounded by:

$$c^{d(u,v)}/X^{1/4}.$$

Now, we'll choose to set our absolute constant $d$ as a function of $c$ (the absolute constant promised by Proposition~\ref{prop:redalwaysred} for $a = 2$) so that $\ln(c)/d = 1/8$ (or $d = 8 \ln(c)$), and recall that $d(u,v)$ is promised to be at most $\frac{\ln(X)}{d}$. Then:

$$c^{d(u,v)}/X^{1/4} \leq c^{\ln(X)/d}/X^{1/4} = e^{\ln(X) \ln(c)/d}/X^{1/4} = e^{\ln(X)/8}/X^{1/4} = X^{-1/8}.$$
}\end{proof}

\begin{proof}[Proof of Theorem~\ref{thm:majority}]
{Lemma~\ref{lem:expect} states that the expected number of \correct nodes is at least $(1/2+\delta-e^{-T/n})n$. Combined with Lemma~\ref{lem:deviation}, if we can show that $\sum_{u,v} \varepsilon^t_{uv} =o(n^2)$, then with probability $1-o(1)$, the realized number of \correct nodes at time $T$ is $(1/2+\delta/2-e^{-T/n})n$, proving the theorem.

Let $A$ denote the set of all $(u,v)$ such that $d(u,v) \geq \frac{\ln n}{8\ln\ln n}$. Let further $B$ denote the set of all $(u,v) \notin A$ such that $\prod_{w \in P(u,v)}\deg(w) \leq X = \sqrt{n}$. Finally, let $C$ denote all remaining pairs.

Observe that as $T \leq \frac{n\ln n}{32 \ln\ln n}$, pairs in $A$ satisfy the hypotheses of Proposition~\ref{prop:long} with $f(n) = \frac{n\ln n}{8\ln\ln n}$ and $k = f(n)/T \geq 4$. Therefore, such pairs have $\varepsilon^T_{uv} \leq 2e^{- \ln(n)/(24 \ln \ln n)}$ (used in the case where $T$ is close to $\frac{n\ln n}{32 \ln \ln n}$), and also {$\varepsilon^T_{uv} \leq 2(\frac{32eT\ln \ln n}{n\ln n})^{\ln n/(8 \ln \ln n)}$}. For sufficiently large $n$, and $T \leq \frac{n\ln^{1-\gamma}}{\ln \ln n}$, the latter term is further upper bounded by $n^{1-\alpha}$, for some constant $\alpha$ which depends on $\gamma$. Let's further update $\alpha := \min\{\alpha, 1/16\}$ for simplifying wrapping up at the end.

Pairs in $B$ satisfy the hypotheses of Lemma~\ref{lem:shortpath1} with $X = \sqrt{n}$, and therefore there are at most $n^{3/2}$/2 such pairs. Finally, let's confirm that pairs in $C$ satisfy the hypotheses of Lemma~\ref{lem:shortpath2} for $X = \sqrt{n}$ for sufficiently large $n$.

For the condition on $d(u,v)$, observe that $\ln(X)/d \geq \ln(n)/(2d)$, while we are guaranteed for pairs in $C$ that $d(u,v) \leq \frac{\ln n}{8\ln \ln n}$, which indeed is $\leq \ln(n)/(2d)$ for sufficiently large $n$. Next, we also need to make sure that the condition on $T$ is satisfied. Here, observe that the condition on $T$ is the strongest when $\prod_{w \in P(u,v)} \deg(w)$ is as \emph{small} as possible ($\sqrt{n}$), and $d(u,v)$ is as \emph{big} as possible ($\frac{\ln n}{8 \ln \ln n}$). So our condition becomes (below the LHS denotes the strictest upper bound on $T/n$ Corollary~\ref{cor:caseThree} might demand, and the RHS is simply the biggest $T/n$ we wish to apply Corollary~\ref{cor:caseThree} for). 

\begin{align*}
e^{b\sqrt{n}^{2\ln \ln n/ \ln n}} \geq \frac{\ln n}{32 \ln \ln n}\\
\Leftrightarrow e^{b e^{\ln n \ln \ln n / \ln n}} \geq\frac{\ln n}{32 \ln \ln n}\\
\Leftrightarrow e^{b e^{\ln \ln n}} \geq\frac{\ln n}{32 \ln \ln n}\\
\Leftrightarrow e^{b \ln n} \geq\frac{\ln n}{32 \ln \ln n}\\
\Leftrightarrow n^b \geq \frac{\ln n}{32\ln\ln n}
\end{align*}

The last inequality clearly holds for sufficiently large $n$. Now, we can simply bound $\sum_{u,v} \varepsilon^T_{uv}$ by summing over $A, B, C$:

$$\sum_{u,v} \varepsilon^T_{uv} = \sum_{(u,v)\in A} \varepsilon^T_{uv} + \sum_{(u,v)\in B} \varepsilon^T_{uv} + \sum_{(u,v)\in C} \varepsilon^T_{uv}$$
$$\leq |A|\cdot 2e^{- \ln(n)/(24 \ln \ln n)} + |B| + |C| \cdot n^{-1/16}$$
$$\leq n^2 \cdot e^{-\ln n / (24 \ln \ln n)},$$}

which holds in all cases, or we get the stronger bound below in the case that $T \leq \frac{n\ln^{1-\gamma}n}{\ln \ln n}$:

$$\sum_{u,v} \varepsilon^T_{uv} = \sum_{(u,v)\in A} \varepsilon^T_{uv} + \sum_{(u,v)\in B} \varepsilon^T_{uv} + \sum_{(u,v)\in C} \varepsilon^T_{uv}$$
$$\leq |A|\cdot 2n^{-\alpha} + |B| + |C| \cdot n^{-1/16}$$
$$\leq n^2 \cdot n^{-\alpha}.$$

\end{proof}

\section{Omitted Proofs From Section~\ref{sec:stable}}\label{app:stable}

\begin{proof}[Proof of Lemma~\ref{lem:stablechildren}]
Assume that $v$ changes her opinion at $t >T$ from $A$ to $B \neq A$, and her previous announcement (of $A$) was at time $t''$. Note that $t'' \geq t'\geq T$ by assumption. By Proposition~\ref{prop:critical}, there must have been some node $x$ that had $C^{t''}(x) \neq B$, but $C^{t}(x) = B$. However, during this entire window, $v$'s announcement is $A$, and during this entire window all of $v$'s children are nearly-finalized with respect to $v$. Therefore, any child of $v$ that changes her announcement during this window necessarily changes it to $A$. Therefore, if $v$ changes opinion to $B$ at $t$, it must be because $x = u$, and therefore we indeed have $C^t(v) = C^t(u)$ as desired.
\end{proof}

\begin{proof}[Proof of Lemma~\ref{lem:stable}]
The proof proceeds by induction. As a base case, consider when $v$ is a leaf. Then every time that $v$ announces it will copy its parent, as long as the announcement of its parent is not $\bot$. If its announcement is $\bot$, then $v$ will simply announce its initial belief. $v$ will finalize once its parent finalizes (which eventually happens with probability $1$ for all nodes). Therefore, $v$ is nearly-finalized with respect to its parent as soon as $v$ announces for the first time, which occurs at $T(v,v)$. So the base case holds.

Now assume that the lemma holds for all children of $v$. Observe first that $T_v >T_x$ for all $x$ that are children of $v$ (this is simply because any descendant $y$ of $x$ is also a descendant of $v$ and has $T(y,x) < T(y,v)$). Therefore, the inductive hypothesis claims that all children of $v$ are nearly-finalized with respect to $v$ at $T_v$. Moreover, observe that $v$ necessarily announces at $T_v$ (because $T_v = T(y,v)$ for some $y$, and $v$ announces at each $T(y,v)$). 

Lemma~\ref{lem:stablechildren} then claims that for any $t>T_v$ during which $v$ changes her announcement, $v$ copies her parent, $u$. We claim that this implies that $v$ is nearly-finalized with respect to $u$. {That is, until $v$ finalizes, $v$ \emph{always} copies her parent.}

Indeed, assume for contradiction that $v$ is not nearly-finalized with respect to $u$ at $T_v$. Then there must exist some time $t > T_v$ where $v$ changes their announcement (because $v$ must not be finalized at $T_v$). Let $t$ be the latest such time. Moreover, there must exist some time $t' \in (T_v,t)$ where $v$ announced such that either $C^{t'}(v) \neq C^{t'}(u) \neq \bot$, or $C^{t'}(v) \neq C^{t'-1}(v)$ and $C^t(u) = \bot$.

A direct application of Lemma~\ref{lem:stablechildren} in fact immediately rules the second case out. So let $A=C^{t'}(v)$, and $B = C^{t'}(u)\neq A$. Note also that both $A$ and $B$ are not $\bot$. But now recall that $v$ announced during time $t'$ and updated to a majority of her neighbors (tie-breaking for $X(v)$). All of her neighbors except for $u$ are nearly finalized with respect to $v$ at $t'$. So there cannot be some time $t''>t'$ where a majority of $v$'s neighbors (tie-breaking for $X(v)$) are $B$! This is because the first such $t''$ would have required some neighbor to flip from $A$ to $B$. It cannot be a child of $v$ because they can only flip to $A$. It cannot be $u$ because $u$ is already announcing $B$. So in fact, $v$ must be finalized at $t'$ (contradicting that $v$ changed their announcement at $t >t'$).
\end{proof}

\begin{proof}[Proof of Theorem~\ref{thm:stable}]
With Lemma~\ref{lem:stable}, we just need to get a bound on $T_v$. We do this by partitioning the $T$ steps into $Y$ epochs of length $T/Y$. We'll assign descendants that are at distance $i$ from $v$ to epoch $Y-i$. 

We'll now observe the following: if it is the case that all descendants of $v$ make at least one announcement during their assigned epoch, then $T_v \leq T$. This is because for any descendant $u$ of $v$, we have an ordered sequence of announcements along the path from $u$ to $v$ terminating by $T_v$. So we just want to bound the probability that any descendant doesn't announce during its assigned epoch.

For a single node, the probability that it doesn't announce during a window of length $T/Y$ is exactly $(1-1/n)^{T/Y} \leq e^{-T/(nY)}$. Taking a union bound over each of the $X$ nodes yields that the probability that any node fails is at most $Xe^{-T/(nY)}$.
\end{proof}

\begin{proof}[Proof of Lemma~\ref{lem:flip}]
Let $t'\in [T_v,t)$ be the most recent announcement of $v$ (we know that such a $t'$ exists, because $v$ announces at $T_v$), and say that $v$ announced $A$. Then $v$ copies the majority of its neighbors at $t'$, so a majority of its neighbors announced $A$. In between $t'$ and $t$, $v$ does not announce again (by definition). Other children of $v$ might announce, but because they nearly-finalized with respect to $v$, they will either keep their announcement, or switch to $A$. $v$'s parent may announce, and could announce either $A$ or $B$. But we know that one of $v$'s children changes their announcement at $t$, and because they are all nearly-finalized with respect to $v$, this announcement will be $A$, pushing the majority of $v$'s neighbors further towards $A$. Maybe $v$'s parent will indeed change their announcement from $A$ to $B$, but all this does is cancel out the child's switch to $A$, maintaining an $A$ majority. No other children of $v$ can possibly switch to $B$ without $v$ switching first, so $v$ will stay $A$ forever.
\end{proof}

\begin{proof}[Proof of Corollary~\ref{cor:counting}]
At any step $t > T_v$, let's look at the current states of $v$'s children. We know that every child is nearly-finalized with respect to $v$, so they are either finalized, or copying $v$ with every announcement. If $\lfloor (\deg(v)-1)/2 \rfloor$ of $v$'s children are finalized, then the corollary statement is satisfied. 

If not, then $v$ has at least $\lceil (\deg(v)-1)/2\rceil +1 >\deg(v)/2$ non-finalized children. If all of these children agree with $v$'s current announcement, then in fact $v$ is finalized. This is because each of these $\lceil (\deg(v)-1)/2 \rceil+1$ children will not change their announcement unless prompted by $v$, and $v$ will not change her announcement unless one of them change (because this constitutes a strict majority of $v$'s neighbors). If $v$ is finalized, then each of these children that agree with $v$ are certainly finalized, and the corollary is again satisfied.

So if $v$ is not finalized, \emph{and} $v$ has fewer than $\lfloor(\deg(v)-1)/2\rfloor$ finalized children, $v$ must have a non-finalized child that disagrees with $v$. Whoever this child is, it is nearly-finalized with respect to $v$. So if it is selected to announce, it will change its opinion to match $v$, causing $v$ to finalize by Lemma~\ref{lem:flip}, and have the desired number of finalized children by the reasoning above. So if $v$ does not yet have the desired number of finalized children by $t$, there is at least one child of $v$ such that if they are selected to announce, $v$ will then have the desired number of finalized children. The probability that this child is \emph{not} selected to announce for all $t \in (T_v,T_v+T]$ is at most $e^{-T/n}$. Note that as $v$ announces (possibly switching their announcement) during this window, the identity of this child may change. But it remains the case that at all steps $t$, either $v$ already has the desired number of finalized children, or there is some child who will cause $v$ to finalize when selected to announce.

For the ``Moreover,\ldots'' part of the statement, we'll use a similar proof to that of Theorem~\ref{thm:stable}. We'll again break down the $T$ steps into epochs of length $T/Y$. This time, we'll put descendants at distance $i$ from $v$ into epoch $i$. By the reasoning further above (which claims that once $v$ is finalized and a child $x$ of $v$ announces, $x$ is for sure finalized), we again conclude that as long as every descendant makes an announcement durings its assigned epoch, all descendants of $v$ are finalized. The probability that any descendant fails to announce during its assigned epoch is $\leq e^{-T/(nY)}$, so a union bound gives that the total failure probability is $Xe^{-T/(nY)}$.
\end{proof}

\section{Omitted Proofs From Section~\ref{sec:wrapup}}\label{app:wrapup}

\begin{proof}[Proof of Proposition~\ref{prop:good}]
First, observe that by Proposition~\ref{prop:redalwaysred} $v$'s parent is safe thru $O(n \ln n)$ with probability $1-o(1)$. By Corollary~\ref{cor:stopping}, the entire process actually terminates by $O(n \ln n)$, with probability $1-o(1)$. Therefore, $v$'s parent is finalized as soon as they announce (except with probability $o(1)$). 

So, the probability that $v$'s parent hasn't announced by $n \cdot \frac{\ln n}{64 \ln \ln n}$ is $o(1)$. By Theorem~\ref{thm:stable}, $v$ is nearly-finalized with respect to its parent by $n\cdot  \frac{\ln n}{64 \ln \ln n}$ with probability $1-o(1)$ (this can be observed by plugging in $T = \frac{n\ln n}{64 \ln \ln n}$, $X = \ln^{O(1)}(n), Y = O(\ln \ln n)$). When both events happen, this in fact means that $v$ is finalized. Finally, Corollary~\ref{cor:counting} implies that with probability $1-o(1)$, the entire subtree rooted at $v$ is finalized by an additional $n\cdot \frac{\ln n}{64 \ln \ln n}$ steps, proving the proposition.
\end{proof}

\begin{proof}[Proof of Proposition~\ref{prop:PA}]
We want to claim that most subtrees rooted at nodes that arrive after $n/\ln n$ and are a child of one of the first $n/\ln n$ are good. We will call the first $n/\ln n$ nodes to arrive, the \emph{early nodes}. For any node $v$, let $s(v)$ be the size of the subtree $H_v$ rooted at $v$ at the end of the process. Consider the process after all the early nodes arrive. Every time an early node $v$ succeeds and gets a child $u$, we call $u$ a \emph{critical node} and the subtree $H_u$ a \emph{critical subtree}, and $v$ is the \emph{source} of $H_u$. 

To make the proofs simpler, we'll view the Preferential Attachment model in the following way, as introduced by Pittel~\cite{Pittel94}. Note that this process is \emph{equivalent} to the standard Preferential Attachment definition given in the body. For every node $v$, immediately upon creation, start a Poisson clock with rate $1$. That is, draw a waiting time $W_v$ from the exponential distribution of rate one, and wait until $W_v$. When the clock ticks, create a new node $u$ that is a child of $v$, and immediately start a Poisson clock with rate $2$. Every time the clock ticks, create a new node (with its own Poisson clock), increase the rate by $1$ (to match the new degree of $v$), and immediately start the clock again. 

We'll refer to the \emph{subtree process at $v$} to be all Poisson processes of all descendents of $v$, and the \emph{node process at $v$} to simply be $v$'s Poisson process. Owing the memoryless nature of Poisson processes (and other properties), Pittel shows that this is equivalent to the standard definition. We repeat the proof below for completeness:

\begin{definition}[Alternate Process for Preferential Attachment Tree] $n$ nodes arrive sequentially, attaching a single edge to a pre-existing node at random proportional to its degree. Specifically:
\begin{itemize}
\item At continuous time $t = 0$, start a Poisson clock for node $v_1$ at rate $r_1 = 1$.
\item Initialize $i = 2$.
\item When any Poisson clock ticks:
\begin{itemize}
\item Create a new node $v_i$. Start a Poisson clock for node $v_i$ at rate $r_i = 1$. 
\item Let $j < i$ denote the index of the clock that ticked. 
\item Add an edge from $v_i$ to $v_j$. Increase the rate $r_j$ by $1$ ($r_j:= r_j+1$). Start a Poisson clock for node $v_j$ at rate $r_j$.
\item Increment $i$ by $1$ ($i:= i+1$).
\end{itemize}
\end{itemize}
\end{definition}

\begin{proposition}[\cite{Pittel94}]
The two definitions for the Preferential Attachment Process are identical.
\end{proposition}
\begin{proof}
The only thing we need to confirm is that every time a node arrives in the alternate process, the node it picks to attach to is drawn proportional to its degree. To see this, consider $t_i$, the (continuous) time of arrival of $v_i$, and $t_{i+1}$, the time of arrival of $v_{i+1}$. Because Poisson processes are memoryless, conditioned on one of them ticking at $t_i$, it is as if all of them were restarted at $t_i$. So, the node selected for $v_{i+1}$ to attach to is simply whichever clock ticks first.

Finally, observe that which clock ticks first is simply whichever waiting time is smallest, and each waiting time is drawn from an exponential distribution of rate equal to the degree. Recall that if $X_1,\ldots, X_k$ are independent exponential random variables of rates $r_1,\ldots, r_k$, then $\ell$ is the argminimum with probability exatly $r_\ell/\sum_j r_j$ for all $\ell$. That is, $\ell$ is the argminimum with probability proportional to its degree, meaning that indeed the next clock to tick is chosen proportional to the degree, and meaning that node $i+1$ attaches to a random node proportional to its degree, as desired.
\end{proof}

What's nice about the alternative view is that in some sense the only correlating feature between what happens in subtrees rooted at different nodes is the continuous time. Let's return now to our proof approach. Our plan is to look at all critical subtrees, and show that every critical subtree is good with probability $1-o(1)$. As all non-early nodes are in a critical subtree, Markov's inequality will suffice to let us conclude the proposition.

First, we'll want to analyze the size of a critical subtree, and show that with high probability, no critical subtree is too large. To this end, we prove several useful lemmas, all of which make use of ideas from~\cite{Pittel94}, and will rely on the following concentration inequality for exponential random variables:

\begin{theorem}[\cite{Janson14}]\label{thm:exp-concentration}
Let $X=\sum_{i=1}^mX_i$ with $X_i\sim Exp(a_i)$ independent, $\mu=\E(X)$ and $a=\min_i a_i$.
\begin{itemize}
\item For any $\lambda\ge1$, $\Pr(X\ge\lambda\mu)\le \lambda^{-1}e^{-a\mu(\lambda-1-\ln \lambda)}$
\item For any $\lambda\le1$, $\Pr(X\le\lambda\mu)\le e^{-a\mu(\lambda-1-\ln \lambda)}$
\end{itemize}
\end{theorem}

The following conceptual fact will be helpful in reasoning about the subtree processes without re-doing the same arithmetic.

\begin{fact}
The subtree process at $v$ is itself a Preferential Attachment process, just starting at a time $t > 0$.
\end{fact}

\begin{lemma}
Let $Y_i$ denote the waiting time in between when the $i^{th}$ node arrives and when the $(i+1)^{th}$ node arrives. Then $Y_i$ is an exponential random variable of rate $2i-1$, and all $Y_i$ are independent.
\end{lemma}
\begin{proof}
The fact that $Y_i$ are independent follows immediately from the memorylessness property of exponential random variables. The waiting time between when the $i^{th}$ node arrives and the $(i+1)^{th}$ is just the minimum of $i$ exponential random variables, with rates $r_1,\ldots, r_i$. The minimum of these is itself distributed according to an exponential distribution of rate $\sum_j r_j = 2i-1$, as desired.
\end{proof}

\begin{lemma}\label{lem:PA1}
For all $k$, $\sum_{i=k}^m Y_i$ is the time from when the $k^{th}$ node arrives until the $m^{th}$. Also:
\begin{itemize}
\item $\mathbb{E}[\sum_{i=k}^m Y_i] \in (\ln(m/k)/4, \ln(m/k))$,
\item $\Pr[\sum_{i=k}^m Y_i < \ln(m/k)/8] \leq e^{-\Omega(k\ln(m/k))}$, and
\item $\Pr[\sum_{i=k}^m Y_i > 2\ln(m/k)] \leq e^{-\Omega(k\ln(m/k))}$.
\end{itemize}
\end{lemma}
\begin{proof}
Observe that each $Y_i$ is an exponential random variable of rate at least $2i-1 \geq k$. Therefore, $\mathbb{E}[\sum_{i=k}^{m} Y_i] = \sum_{i=k}^m \frac{1}{2i-1} \in (\ln(m/k)/4, \ln(m/k))$, and we can apply Theorem~\ref{thm:exp-concentration} with $\lambda = 1/2, a = k, \mu \geq \ln(m/k)/4$ to get the first desired bound. The second one follows from $\lambda = 2, a = k, \mu \leq \ln(m/k)$.
\end{proof}

Below, we improve the second bound in the case of small $k$.
\begin{corollary}When $k = 1$, we improve Lemma~\ref{lem:PA1} to:
\begin{itemize}
\item $\mathbb{E}[\sum_{i=1}^m Y_i] \in (\ln(m)/4, \ln(m))$, and
\item $\Pr[\sum_{i=1}^m Y_i < \ln(m)/16] \leq e^{-\Omega(\sqrt{m} \ln m)}$.
\end{itemize}
\end{corollary}
\begin{proof}
Observe that $\sum_{i=1}^m Y_i \geq \sum_{i=\sqrt{m}}^m Y_i$. We'll therefore apply Lemma~\ref{lem:PA1} with $k = \sqrt{m}$ and conclude that the probability that $\sum_{i = \sqrt{m}}^m Y_i < \ln(\sqrt{m})/8$ is at most $e^{-\Omega(\sqrt{m}\cdot \ln (m))}$. 
\end{proof}

With these, we may now claim the following:
\begin{proposition}\label{prop:criticalsubtrees}
With probability at least $1-e^{-\Omega(\ln^{15}(n) \ln \ln n)}$, every critical subtree has at most $\ln^{32}(n)$ nodes. 
\end{proposition}
\begin{proof}
First, we can conclude from Lemma~\ref{lem:PA1} that except with probability $e^{-\Omega(n\ln \ln n / \ln n)}$ (plugging in $k = n/\ln n$, $m = n$), the time it takes to go from $n/\ln n$ nodes until $n$ nodes (and the process terminates) is at most $2\ln \ln n$. Similarly, for any non-early node, the probability that it achieves $\ln^{32}(n)$ descendents before $2\ln \ln n$ is at most $e^{-\Omega(\ln^{16}(n) \ln \ln n)}$. Taking a union bound over the (at most) $n$ critical subtrees, we get that the probability that any non-early node has more than $\ln^{32}(n)$ descendents is $e^{-\Omega(\ln^{15}(n) \ln \ln n)}$. 
\end{proof}

At this point, we've shown that with very high probability, no critical subtree is too large, which completes step one of the proof. Now, we wish to analyze the degree of the source of critical subtrees. We'll do this in two steps. First, we'll claim that if the root of a critical subtree arrives sort-of early (before $n/\sqrt{\ln n}$), then with good probability that root's source gains $\ln^{\Omega(1)}(n)$ further children by the end of the process. We'll then claim that if the root arrives late, then with good probability the root it attaches itself to already has degree $\ln^{\Omega(1)}(n)$. 

To this end, we'll denote $Z^v_i$ to be the waiting time between when node $v$ gets its $(i-1)^{th}$ child and its $i^{th}$. 

\begin{lemma}\label{lem:PAdegree}
For all $k, m$, we have:
\begin{itemize}
\item $\mathbb{E}[\sum_{i=k}^m Z^v_i] = \ln(m/k) + O(1)$.
\item $\Pr[\sum_{i=k}^m Z^v_i > 2 \ln(m/k)] \leq e^{-\Omega(k \ln (m/k))}$.
\end{itemize}
\end{lemma}
\begin{proof}
Observe that each $Z^v_i$ is an exponential random variable of rate $i$. As such, $\mathbb{E}[\sum_{i=k}^mZ^v_i] = \sum_{i=k}^m 1/i = \ln(m/k) + O(1)$. Again, directly plugging into Theorem~\ref{thm:exp-concentration} with $a = k$, $\lambda = 2$, and $\mu = \ln(m/k) + O(1)$ we get the desired bound.
\end{proof}

\begin{corollary}\label{cor:kindofearly}
For any node $v$ that arrives before the $(n/\sqrt{\ln n})^{th}$ arrival, no matter its degree at the arrival of the $(n/\sqrt{\ln n})^{th}$ arrival, with probability at least $1- e^{-\Omega(\ln\ln n)}$, $v$ gets $\ln^{1/32}(n)$ new neighbors before the $n^{th}$ arrival.
\end{corollary}
\begin{proof}
First, again conclude immediately from Lemma~\ref{lem:PA1} (plugging in $k = n/\sqrt{\ln n}$, $m = n$) that except with probability $e^{-\Omega(n\ln\ln n/\sqrt{\ln n})}$, the time it takes to go from $n/\sqrt{\ln n}$ nodes until $n$ nodes is at least $\ln\ln (n)/16$. Similarly, Lemma~\ref{lem:PAdegree} (observing that the worst case is when $v$ has no children yet at $n/\sqrt{\ln n}$) claims that except with probability $e^{-\Omega(\ln \ln n)}$, $v$ has at least $\ln(n)^{1/32}$ children by time $\ln\ln(n)/16$. Taking a union bound completes the proof.
\end{proof}

This allows us to conclude that if a critical subtree attaches to its source \emph{before} $n/\sqrt{\ln n}$, then with high probability that source gets $\ln^{\Omega(1)}(n)$ neighbors by termination. We now consider the case that a critical subtree attaches to its source \emph{after} $n/\sqrt{\ln n}$. To this end, we'll denote $Z^S_i$ to be the waiting time between when a \emph{set of nodes} $S$ gets its $(i-1)^{th}$ direct descendent and its $i^{th}$ (where by a direct descendent, we mean a node not in $S$ attached to a node in $S$). 

\begin{lemma}\label{lem:PAsets}
Let $S$ be the first $s$ nodes to arrive in preferential attachment. Then for $m \geq s$:\footnote{The assumption that $m \geq s$ is for simplicity of exposition, as this is the only case we need. The ideas in the lemma hold for arbitrary $m$ but would require keeping track of the difference between $m$ and $m+s$.}
\begin{itemize}
\item $\mathbb{E}[\sum_{i=1}^m Z^S_i] = \ln(\frac{m}{2s}) + O(1)$.
\item $\Pr[\sum_{i=1}^m Z^S_i > 2 \ln(\frac{m}{2s})] \leq e^{-\Omega(s \ln (m/s))}$.
\end{itemize}
\end{lemma}
\begin{proof}
Observe again that $Z^S_i$ is an exponential random variable of rate $2s-2+i$. So $\mathbb{E}[\sum_{i=1}^m Z^S_i] = \sum_{i=1}^m \frac{1}{2s-2+i} = \ln(\frac{m}{2s}) + O(1)$. For the second bullet, we again plug directly into Theorem~\ref{thm:exp-concentration} ($\lambda = 2, a = s, \mu = \ln(\frac{m}{2s})+O(1)$) to get the desired bound.
\end{proof}

\begin{corollary}\label{cor:PAsets}
Let $S$ be the first $s = n/\ln n$ nodes to arrive in preferential attachment. Then with probability at least $1-e^{\Omega(n\ln \ln n/\ln n)}$, by the time $n/\sqrt{\ln n}$ nodes have arrived, the set $S$ has at least $s\cdot \ln^{1/32}(n)$ direct descendents. 
\end{corollary}
\begin{proof}
Again, immediately from Lemma~\ref{lem:PA1} (plugging in $k = n/\ln n$, $m = n/\sqrt{\ln n}$), we conclude that except with probability $e^{-\Omega(n\ln\ln n/\ln n)}$, the time it takes to go from $n/\ln n$ nodes until $n/\sqrt{\ln n}$ nodes is at least $\ln\ln(n)/16$. Similarly, Lemma~\ref{lem:PAsets} concludes (plug in $m/s = \ln^{1/32}(n), s = n/\ln n$) that except with probability $e^{-\Omega(n\ln \ln n/\ln n)}$, $S$ achieves at least $s\cdot \ln^{1/32}(n)$ descendents by $\ln \ln (n)/16$.
\end{proof}

Corollary~\ref{cor:PAsets} lets us conclude that except with probability $e^{-\Omega(n\ln \ln n/\ln n)}$, the \emph{average degree} of the early nodes is at least $\ln^{1/32}(n)$ by the time the $(n/\sqrt{\ln n})^{th}$ node arrives. Therefore, we can conclude the following:

\begin{corollary}\label{cor:lateattach}
Let $v$ be a critical node arriving after $n/\sqrt{\ln n}$. Then (conditioned on being a critical node), except with probability $\ln^{-1/64}(n)$, $v$ attaches itself to a source with degree at least $\ln^{1/64}(n)$. 
\end{corollary}
\begin{proof}
Conditioned on being a critical node, $v$ picks a source proportional to degrees. By Corollary~\ref{cor:PAsets}, except with probability $e^{-\Omega(n\ln \ln n/\ln n)}$, the total sum of degrees is at least $s \cdot \ln^{1/32}(n)$. If we sum the degrees of all early nodes with degree at most $\ln^{1/64}(n)$, we get at most $s \cdot \ln^{1/64}(n)$. Therefore, when we pick a node proportional to degrees, the probability that such a node is picked is at most $\ln^{-1/64}(n)$. 
\end{proof}

Combining Corollaries~\ref{cor:kindofearly} and~\ref{cor:lateattach} we are able to conclude now that every critical subtree attaches itself to a source of degree at least $\ln^{1/64}(n)$ with high probability. Finally, we just need to bound the diameter. To this end, we'll make black-box use of a theorem of~\cite{DommersHH09}:

\begin{theorem}[\cite{DommersHH09}]\label{thm:diameter}
Let $G(k)$ be a tree formed by the Preferential Attachment model with $k$ nodes. Then, with probability $1-o(1)$, the diameter of $G(k)$ is $\Theta(\ln k)$.
\end{theorem}

To wrap everything up, we now do the following. First, condition on the set of critical nodes. Then, for each critical node $v$, condition on the set $S_v$ of nodes in $v$'s critical subtree. But don't yet determine $v$'s source (aside from that it is an early node), and don't yet determine the structure of $v$'s critical subtree. Note that indeed where $v$ attaches is independent of the subtree structure below it, as is the size of the subtree. The former claim is simply because no matter where $v$ attaches itself, the subtree below it grows like an independent preferential attachment process. The latter claim is because no matter how the subtree below $v$ is oriented, the probability that a node joins that subtree is proportional to the size of that subtree, independent of the structure (because only the sum of degrees matters).

We already have from Proposition~\ref{prop:criticalsubtrees} that except with probability $e^{-\Omega(\ln^{15}(n) \ln \ln n)}$, every critical subtree has at most $\ln^{32}(n)$ nodes, so we will proceed under this assumption. Next, we want to claim that with high probability, a $1-o(1)$ fraction of late nodes are in a critical subtree with diameter $O(\ln \ln n)$. 

\begin{proposition}\label{prop:lowdiam}
Except with probability $1-o(1)$, a $1-o(1)$ fraction of late nodes are in a critical subtree with diameter $O(\ln \ln n)$. 
\end{proposition}
\begin{proof}
Note that there are a total of $r=n-n/\ln n$ non-early nodes. For each critical subtree rooted at $v$ of size $D_v$ (but whose structure is yet to be determined), define an indicator random variable $W_v$ which is $D_v/\ln^{32}(n)$ if the diameter exceeds $c \cdot \ln \ln n$, where $c$ is the constant in the big-Oh upper bound promised by Theorem~\ref{thm:diameter}, and $0$ otherwise. Then by Theorem~\ref{thm:diameter}, we know that {$\mathbb{E}[\sum_v W_v] = o(r/\ln^{32} n)$}. Also, all $W_v$ are independent. Therefore, the additive Chernoff bound asserts that the probability that $\sum_v W_v$ exceeds its expectation by more than $n^{3/4}= o(r/\ln^{32}(n))$ is at most $e^{-\Omega(n^{1/3})}$ {(because there are at least $r/\ln^{32}(n)$ random variables being summed, as the sum of $D_v$s is at least $r$ and each $D_v$ is at most $\ln^{32} n$)}. If $\sum_v W_v$ doesn't exceed its expectation, which is $o(r/\ln^{32}(n))$, by more than $o(r/\ln^{32}(n))$, then the total number of nodes in a critical subtree with diameter exceeding $c \cdot \ln \ln n$ is {$o(r)$}, completing the proof.
\end{proof}

Now we have that except with probability $o(1)$, every critical subtree has at most $\ln^{32}(n)$ nodes, and a $1-o(1)$ fraction of non-early nodes are in a critical subtree of diameter $O(\ln \ln n)$. We now proceed under both of these assumptions, and look at any node that is in a critical subtree of diameter $O(\ln \ln n)$. If the root of this tree attaches itself to a source of degree at least $\ln^{1/64}(n)$, then the entire subtree is good. We have shown that this occurs with probability at least $1-\ln^{-1/64}(n)$. Therefore, the expected number of nodes that are in a critical subtree of diameter $O(\ln \ln n)$ but whose root attaches to a source of \emph{low degree} is at most {$r \cdot \ln^{-1/64}(n)$. By Markov's inequality, the number of such nodes is at most $r \cdot \ln^{-1/128}(n)$ }with probability at least $1-\ln^{-1/128}(n)$.

So in summary: we first argued that with high probability, \emph{every} critical subtree has at most the desired number of nodes. Conditioned on this, we argued further that with high probability, almost all non-early nodes are in critical subtrees with the desired diameter. Finally, we argued that conditioned on this, with high probability almost all non-early nodes are further in critical subtrees whose source has the desired degree. Together, this claims what with probability $1-o(1)$, $n-o(n)$ nodes are in good subtrees (as we've argued that this is the case for a $1-o(1)$ fraction of the $n-n/\ln n$ non-early nodes). 

To wrap up, we simply take another union bound with the $o(1)$ failure probability promised by Theorem~\ref{thm:diameter} that the diameter of the entire graph is not $O(\ln n)$. 
\end{proof}

\begin{proof}[Proof of Theorem~\ref{thm:PA}]
Simply combine Propositions~\ref{prop:good} and~\ref{prop:PA}. Together, they say that with probability $1-o(1)$, only $o(n)$ nodes are \emph{not} in good subtrees. Moreover, the expected number of nodes that \emph{are} in good subtrees but do \emph{not} finalize by $n\cdot\frac{\ln n}{32 \ln \ln n}$ is also $o(n)$. Therefore, Markov's inequality alone suffices to claim with that with probability $1-o(1)$, $n-o(n)$ nodes have finalized by $n\cdot\frac{\ln n}{32 \ln \ln n}$.

Theorem~\ref{thm:majority} claims that with probability $1-o(1)$, $n/2 + \delta n/4$ nodes have a \correct annoucement. When both of these conditions hold, at most $o(n)$ nodes can possibly change their future announcement from this point, and therefore the \correct majority holds until termination.
\end{proof}

\begin{proof}[Proof of Proposition~\ref{prop:countingMary}]
Corollary~\ref{cor:height} allows us to conclude that such nodes are nearly-finalized with respect to their parent at $\frac{n\ln n}{64 \ln \ln n}$ with probability $1-o(1)$. Next, we want to claim that most of the descendents of such nodes are finalized by $\frac{n \ln n}{64 \ln \ln n}$.

We'll again consider epochs of time passing from $v$ towards its descendents. We already know that with probability $1-o(1)$, $v$ is nearly-finalized with respect to its parent (plugging in $X = 2M^h$, $T = \frac{n\ln n}{64 \ln \ln n}$, $Y = h$ to Theorem~\ref{thm:stable}), so assume that this holds. We then want to take $T = \frac{n\ln n}{64 (\ln \ln n)^2}$ and apply Corollary~\ref{cor:counting}. This immediately lets us conclude that with probability $1-e^{- \frac{\ln n}{64 (\ln \ln)^2}}$, $v$ has $\lfloor M/2\rfloor \geq M/3$ finalized children. 

Now, for each of $v$'s non-finalized children, we'll apply Corollary~\ref{cor:counting} again with the same choice of $T$. We'll continue this process recursively until we reach the bottom. 

So, the probability that a single non-finalized descendent does not have the desired number of finalized children by the end of its prescribed epoch is at most $e^{-\frac{\ln n}{64 (\ln \ln n)^2}}$. Taking a union bound over all $2M^h$ nodes gives that with probability at most $2M^h e^{-\frac{\ln n}{64 (\ln \ln n)^2}}$ is there \emph{any} failure, and as $M \leq \ln n$, $h = \ln \ln n$, this entire bound is $o(1)$.

Now, we want to further use Corollary~\ref{cor:counting} to conclude that once a node is finalized, all of its descendants are finalized by the end of the epochs. Let now $x$ be some node that is finalized by its intended epoch $i$. Then $x$ is of height $h-i$ and is a $(2M^{h-i},i)$-leaf. Corollary~\ref{cor:counting} therefore immediately implies that the probability that any of $x$'s descendents are not finalized by the end of an additional $(h-i)\frac{\ln n}{64 (\ln \ln n)^2}$ steps is at most $2M^{h-i} e^{-\frac{\ln n}{64 (\ln \ln n)^2}}$. Taking a further union bound over all $2M^h$ nodes each with failure probability at most $e^{-\frac{\ln n}{64 (\ln \ln n)^2}}$ (to announce during its corresponding epoch), we see that this entire failure probability is again $o(1)$. 

So let's recap what we have now. First, we have an $o(1)$ failure probability that $v$ is not nearly-finalized by $n\cdot \frac{\ln n}{64 \ln \ln n}$. Next, we have an $o(1)$ failure probability that \emph{any} non-finalized node does not have at least $M/3$ finalized children by the end of its prescribed epoch. Finally, we have an $o(1)$ failure probability that \emph{any} node with an ancestor who was finalized during their prescribed epoch is not finalized by the end of the entire process (which is exactly $n \cdot \frac{\ln n}{32 \ln \ln n}$). So except with probability $o(1)$, at least an $1/3$ fraction of the remaining descendents get a finalized ancestor in each epoch, and is itself finalized by the end.

The final step is just a simple counting: if the number of remaining descendents without a finalized ancestor shrinks by a $2/3$ factor, and there are $h$ epochs, then the number of unfinalized nodes at the end is at most a $(2/3)^h$ fraction of the initial set.
\end{proof}

\begin{proof}[Proof of Theorem~\ref{thm:Mary}]
For $M \leq \ln n$ there's not much left to wrap up. We just need to count the number of nodes that we've just claimed are finalized by $n\cdot \frac{\ln n}{32 \ln \ln n}$. Observe first that there are $\geq M^H$ nodes of height $\leq \ln \ln n$, and at most $2M^{H-\ln \ln n}$ nodes of height $\geq \ln \ln n$ (where $H$ denotes the height of the root). Therefore, a $1-o(1)$ fraction of all nodes are of height $\leq \ln \ln n$. Moreover, Proposition~\ref{prop:countingMary} proves that a $1-o(1)$ fraction of such nodes are finalized by $n\cdot\frac{\ln n}{32 \ln \ln n}$. Therefore, we conclude that only $o(n)$ nodes in the entire graph are not finalized by $n\cdot\frac{\ln n}{32\ln\ln n}$, except with probability $o(1)$. Theorem~\ref{thm:majority} claims that with probability $1-o(1)$, $n/2 + \delta n/4$ nodes have a \correct annoucement. When both of these conditions hold, at most $o(n)$ nodes can possibly change their future announcement from this point, and therefore the \correct majority holds until termination.

For $M \geq \ln n$, the argument is actually simpler: all pairs $(u,v)$ of nodes contain some node $x$ on $P(u,v)$ that has degree at least $\ln n$. By Proposition~\ref{prop:redalwaysred}, $x$ is safe thru $n\ln^2 n$ even ignoring $S_x$ with probability $1-o(1)$. Therefore, Lemma~\ref{lem:shortpath3} guarantees that $\varepsilon_{uv}^{n\ln^2 n} = o(1)$, and Lemma~\ref{lem:influence} guarantees that there's a \correct majority with probability $1-o(1)$ at $n\ln^2 n$. However, the diameter is $O(\ln n)$, and therefore Corollary~\ref{cor:stopping} guarantees that the entire process stabilizes by $n\ln^2 n$ with probability $1-o(1)$. So taking a union bound, we see that with probability $1-o(1)$ the process has stabilized in a \correct majority by $n\ln^2 n$.
\end{proof}

\end{document}